\documentclass[11pt]{article}

\usepackage[margin=1in]{geometry}
\usepackage{amsmath,amssymb,amsthm,mathtools,bbm}
\usepackage{thmtools}
\usepackage{enumitem}
\usepackage{hyperref}
\hypersetup{hidelinks}

\makeatletter
\AddToHook{env/algorithmic/before}{\def\@currentcounter{ALG@line}}
\makeatother
\usepackage[capitalize,nameinlink]{cleveref}
\crefalias{ALG@line}{line}

\usepackage{microtype}
\usepackage{xcolor}
\usepackage{braket}
\usepackage[final]{showlabels}
\usepackage{algorithm}
\usepackage[noend]{algpseudocode}
\usepackage{comment}
\usepackage{showlabels}
\usepackage{mdframed}
\usepackage{MnSymbol}
\usepackage[style=alphabetic,natbib=true,maxcitenames=2, maxbibnames=10,backend=biber]{biblatex}
\newcommand{\ip}[2]{\langle{#1},{#2}\rangle}

\newtheorem{theorem}{Theorem}[section]
\newtheorem{lemma}[theorem]{Lemma}

\newtheorem{definition}[theorem]{Definition}
\newtheorem{claim}[theorem]{Claim}

\crefname{definition}{definition}{definitions}

\newcommand{\Herm}{\textnormal{Herm}}
\newcommand{\LS}{\textsf{LocalSearch}}
\newcommand{\Cover}{\textsf{Cover}}
\newcommand{\dist}{\textnormal{dist}}
\newcommand{\Sep}{\textnormal{Sep}}
\newcommand{\Sch}{\mathsf{Sch}}
\newcommand{\Tr}{\textnormal{Tr}}
\newcommand{\C}{\mathbb{C}}
\newcommand{\cA}{\mathcal{A}}
\newcommand{\cC}{\mathcal{C}}
\newcommand{\Reals}{\mathbb{R}}
\newcommand{\cD}{\mathcal{D}}

\newcommand{\cT}{\mathcal{T}}
\newcommand{\cH}{\mathcal{H}}
\newcommand{\cK}{\mathcal{K}}
\newcommand{\cV}{\mathcal{V}}
\newcommand{\cS}{\mathcal{S}}
\newcommand{\cL}{\mathcal{L}}
\newcommand{\cG}{\mathcal{G}}
\newcommand{\cR}{\mathcal{R}}
\newcommand{\cN}{\mathcal{N}}

\newcommand{\Naturals}{\mathbb{N}}
\newcommand{\wt}{\widetilde}
\newcommand{\poly}{\textsf{poly}}
\newcommand{\OPT}{\textsc{Opt}}
\newcommand{\conv}{\textnormal{conv}}

\title{Polynomial-Time Algorithms for Nuclear Tensor Norms and Multipartite Quantum Separability}

\author{Martino Bernasconi \and Giulio Malavolta}
\author{
\begin{tabular}{cc}
& \\
{Martino Bernasconi}\thanks{Martino Bernasconi was supported by an ERC grant (Project 101165466 — PLA-STEER).} & {Giulio Malavolta}\thanks{Giulio Malavolta was supported by the European Research Council through an ERC Starting Grant (Grant agreement No.~101077455, ObfusQation) and by the Deutsche Forschungsgemeinschaft (DFG, German Research Foundation) under Germany's Excellence Strategy - EXC 2092 CASA – 390781972.}\\
\small{Bocconi University} & \small{Bocconi University}\\
{\textcolor{black}{\small\texttt{martino.bernasconi@unibocconi.it}}} & %
{\textcolor{black}{\small\texttt{giulio.malavolta@unibocconi.it}}}
\end{tabular}
}

\date{}

\begin{document}
\maketitle

\begin{abstract}
    We study a multilinear optimization problem for tensors with bounded Frobenius norm, where each of the $k$ local factors of dimension $d$ is chosen from a convex set. We give a deterministic algorithm running in time $d^{O(k)}$ for constant additive approximations. As applications, we obtain the first polynomial-time algorithms at constant accuracy for weak membership in the nuclear-norm unit ball of high-order tensors and for multipartite quantum separability in Frobenius norm. Our approach views tensor optimization as a cooperative multiprover game and combines this perspective with a recursive spectral compression procedure. We then show that related ideas admit a quantum implementation when the input state is given through copies, rather than an explicit classical description. In this setting, we obtain an algorithm to solve Frobenius-norm separability testing using $\poly(k)$ copies and $\poly(k\log d)$ time.
\end{abstract}

\newpage
\tableofcontents

\newpage
\section{Introduction}

Let $\cH_1,\ldots, \cH_k$ be finite-dimensional real Hilbert spaces, each of dimension at most $d$. For each $j\in[k]$ let $\cK_j\subseteq B_{\cH_j}$ be a convex set which admits a $\poly(d)$ linear optimization algorithm and let
\[
\cC=\conv\left\{x_1\otimes\cdots\otimes x_k:x_j\in \cK_j\right\}.
\]
We are interested in the following \emph{weak membership problem}: given $M\in\cH_1\otimes\cdots\otimes \cH_k$ decide whether 
\[
  M\in\cC,\quad\text{or}\quad \dist(M,\cC) = \inf_{M'\in \cC}\|M-M'\| \geq \epsilon,
\]
where $\|\cdot\|$ is the norm induced by the inner product in $\cH_1\otimes \cdots\otimes \cH_k$. Besides being a natural algorithmic task on its own, this general formulation specializes to well-known computational problems, which we discuss next.

\paragraph{Problem I: Nuclear Tensor Norms.}
When $\mathcal{K}_j$ is the unit ball of $\Reals^d$, then the convex hull of unit-norm rank-one tensors $\mathcal{C}$ corresponds to all tensors with (real) \emph{nuclear norm} at most $1$ \citep{friedland2018nuclear}, where the nuclear norm of a tensor $\|\cdot \|_{*}$ is defined as
\[
    \|M\|_{*} = \inf\Bigg\{\sum_{i=1}^r|\lambda_i|:M = \sum_{i\in[r]}\lambda_i (x_{1,i}\otimes\cdots \otimes x_{k,i}),\|x_{j,i}\|=1, x_{j,i}\in \Reals^d,r\in \Naturals, \lambda_i\in \Reals\Bigg\}.
\]
Note that linear optimization over $v\in \Reals^d$ reduces to computing $\|v\|$, which is of course polynomial-time. Originating from ideas of \citet{grothendieck1955produits} and \citet{schatten1950theory}, the computational problem of weak membership in the nuclear ball was introduced by \citet*{friedland2018nuclear}, who proved NP-hardness of nuclear-norm evaluation for real tensors of order $k\ge3$ and complex tensors of order $k\ge4$.
They also proved NP-hardness of nuclear-ball weak membership when the accuracy is allowed to depend on the instance, already for fourth-order tensors satisfying additional positivity and symmetry restrictions. \citep{hillar2013most,liu2007complexity} established NP-hardness at inverse-polynomial Frobenius accuracy for $k=3$ over both fields. On the other hand, for fixed $k$ and $\epsilon>0$, the randomized additive spectral-norm algorithm of \citet*{fernandez2005tensor}, yields, as a consequence, polynomial-time weak membership at Frobenius accuracy $\epsilon$. Other works provide polynomial-time algorithms for computing nuclear norms with a dimension-dependent multiplicative approximation factor \citep{hu2015relations,li2016bounds,friedland2018nuclear,
chen2020partitions,qi2022spectral,he2023approximating}. 

\paragraph{Problem II: Quantum Separability.}
Conversely, when $\cK_j=\cD_d$ is the set of density matrices in dimension $d$ and $\cH_j=\Herm(\C^d)$: 
\[
\mathcal{C} =\Sep(d,k):=\conv\left\{\rho_1\otimes\cdots\otimes\rho_k:\rho_1,\ldots,\rho_k\in\cD_d\right\},
\]
is the set of separable $k$-partite states. Since linear optimization over $\cD_d$ corresponds to the computation of the largest eigenvalue, which is again polynomial-time, our problem specializes to the Frobenius-norm \emph{quantum separability} problem \citep{gurvits2003classical}. This problem is known to be NP-hard when $\epsilon = \exp(-\Omega (d))$~\citep{gurvits2003classical} and even $\epsilon=1/\poly(d)$~\citep{gharibian2010strong}, already for $k=2$. When distance is measured in \emph{trace norm} instead, \citet*{harrow2013testing} showed evidence that even a polynomial-time algorithm with \emph{constant} additive error would solve $3$-SAT in time $\exp(O(\sqrt n\,\mathrm{polylog}\,n))$. On the positive side \citet*{brandao2011quasipolynomial} (see also~\citep{shi2015epsilon}) gave a quasi-polynomial algorithm, $\exp(O(\eta^{-2}\log^2 d))$, for bipartite weak membership in the LOCC norm and in the Frobenius norm, extended by \citet*{brandao2017quantum} to $k$ parties, in time $\exp(O(k^4\eta^{-2}\log^2d))$, for the one-way LOCC norm and a multipartite variant of the Frobenius norm in time $\exp(O(18^{k/2}k^4\eta^{-2}\log^2d))$. Very recently, the bipartite version of the problem was shown to be solvable in polynomial time for constant $\epsilon >0$ and distance in Frobenius norm \citep{malavolta2026quantum,jeronimo2026optimal}. While these results plausibly generalize to a constant $k$, independent of $d$, a polynomial-time algorithm for the general case is, to the best of our knowledge, currently open.

\subsection{Our Results}\label{sec:ov-results}

In this work we propose a deterministic polynomial-time algorithm for the general problem stated in the introduction, for any $k$ and $d$, in the \emph{constant} accuracy $\epsilon$ regime. Specializing it to weak membership in the unit ball of the nuclear norm, we obtain the first polynomial-time algorithm for general tensors as an immediate consequence. The algorithm runs in time polynomial in the size of the input tensor, $d^k$.

\begin{theorem}[Nuclear Norm; Informal]\label{thm:ov-nuc}
Let $K_{*}=\{M\in(\Reals^d)^{\otimes k}:\|M\|_{*}\le1\}$. For every $\epsilon>0$ there is a deterministic algorithm that, given $M\in(\Reals^d)^{\otimes k}$, decides whether $M\in K_{*}$ or $\dist(M,K_{*})\ge\epsilon$ in time $d^{\,O_\epsilon(k)}$.
\end{theorem}
If we instead specialize the general problem to density matrices, we obtain the first polynomial-time algorithm for the \emph{multipartite} weak membership problem in the Frobenius norm. This improves upon recent works \citep{malavolta2026quantum,jeronimo2026optimal} that obtained an analogous result for $k=2$. Here the input size is an explicit $d^k \times d^k$ density matrix.

\begin{theorem}[Quantum Separability; Informal]\label{thm:ov-sep}
For every $\epsilon>0$ there is a deterministic algorithm that, given $\rho\in\Herm((\C^d)^{\otimes k})$, decides whether $\rho\in\Sep(d,k)$ or $\dist(\rho,\Sep(d,k))\ge\epsilon$ in time $d^{\,O_\epsilon(k)}$.
\end{theorem}
Our algorithm relies on the fact that we can reduce the weak-membership problem to a dual optimization problem in which the tensor norm is guaranteed to be bounded. Akin to the idea of regularity, this lets us approximate the tensor with only a few spectral components and limit the information factor $j$ needs about the optimal solution on all other factors. Then each factor can ``guess'' the optimal solution and optimize its own factor. This is the fundamental idea behind our results. However, this simple strategy is not enough to obtain a truly polynomial-time algorithm, and it requires several technical improvements (elaborated below).

Next, motivated by applications in quantum information, we also consider the question of multipartite quantum separability in the \emph{property testing} setting \citep{badescu2020lower}, i.e., where the inputs are copies of an unknown quantum state $\rho\in\cD((\C^d)^{\otimes k})$. Here we obtain a much stronger result, namely an algorithm to solve Frobenius-norm weak membership where the number of copies is \emph{independent of $d$} and whose runtime is polynomial in $k\log d$, rather than $kd$. To the best of our knowledge, this is the first dimension-free test for weak membership in the set of separable states. Prior work obtained comparable dimension-free guarantees for the simpler task of \emph{state certification} \citep{buadescu2019quantum} or product state certification but restricted to \emph{pure states} \citep{soleimanifar2022testing,beckey2026optimal}.

\begin{theorem}[Quantum Separability from Copies; Informal]\label{thm:ov-quantum}
For every $\epsilon>0$ there is a quantum algorithm that, given copies of an unknown $\rho\in\cD((\C^d)^{\otimes k})$, distinguishes $\rho\in\Sep(d,k)$ from $\dist(\rho,\Sep(d,k))\ge\epsilon$ with high probability, using $k^{O_\epsilon(1)}$ copies of $\rho$ and total time $k^{O_\epsilon(1)}\cdot\mathrm{polylog}(d)$.
\end{theorem}
We highlight that for the analogous problem with the \emph{trace norm} there is a lower bound of $\Omega(d^{2})$ on the number of copies~\citep{badescu2020lower}. Other copy lower bounds \citep{beckey2025product,akresh2026tight} apply also to the Frobenius norm but assume that the measurements act on one copy at a time, or consider measurements that remain positive after partial transposition. Our algorithm bypasses these barriers by (necessarily) using coherent measurements across copies.

The proof relies on two main ingredients: a copy-efficient maximization of $\Tr((\rho-\sigma)\pi)$ over pure product states $\pi$, with $\rho$ given by copies and $\sigma$ an explicit mixture of product states, which extends the product-state learner of~\citep{bakshi2025learning} to qudits and to signed objective functions; and a reduction, performed on copies, from local dimension $d$ to local dimension $O(k^2/\epsilon^4)$ at a cost of $\mathrm{polylog}(d)$ time, using tools from representation theory.

\subsection{Technical Outline I: Classical Algorithms}\label{sec:ov-classical}

The heart of our classical algorithm is the multilinear optimization problem
\begin{equation}\label{eq:ov-opt}
\OPT(M):=\max_{\rho_1\in\cK_1,\ldots,\rho_k\in\cK_k}
\ip{M}{\rho_1\otimes\cdots\otimes\rho_k},
\qquad \|M\|\le1,
\end{equation}
where each $\cK_j \subseteq B_{\cH_j}$ is a convex set of dimension $d$, and admits a linear optimization oracle. This is precisely the \emph{dual} problem of the general problem stated in the introduction, so a solution for the optimization problem implies an algorithm for the weak membership problem, by a standard application of Frank--Wolfe~\citep{frank1956algorithm}. In this outline, we concentrate on approximating \eqref{eq:ov-opt} in polynomial-time, up to a constant additive factor. It is going to be convenient to view our optimization problem as a \emph{cooperative game}: player $j$ chooses a strategy $\rho_j\in\cK_j$, and the players jointly maximize the payoff in~\eqref{eq:ov-opt}.

\paragraph{The Bipartite Case.}
Consider first $k=2$, viewing $M \in \cH_1 \otimes \cH_2$ as a linear map from $\cH_1$ to $\cH_2$.
Suppose for a moment that the first player is given the optimal solution. Then it could simply communicate that solution to the second player, who then only has to optimize a linear function. To reduce the size of the ``message'' of the first player, we prove that the first player's optimal solution affects the objective only through a limited communication channel, whose bandwidth is controlled by the tensor's small Frobenius norm. Thus, instead of having the first player communicating anything, we can now simply let the second player \emph{guess} the message that would be generated by an optimal solution.

More concretely, fix a threshold $\eta>0$, and let $U:\Reals^r\to\cH_1$ be an isometry onto the subspace spanned by the input singular directions with singular values greater than $\eta$. Since $\|M\|\le1$, there are at most $r\le\eta^{-2}$ such directions. The discarded part has operator norm at most $\eta$, so
\begin{equation}\label{eq:ov-onecut}
\big|\ip{M}{x\otimes y}-\ip{M}{UU^{\dagger}x\otimes y}\big| \le\eta\,\|x\|\,\|y\|,\quad\forall x\in\cH_1,y\in\cH_2.
\end{equation}
Thus, up to error $\eta$, the payoff is $\ip{M}{Up\otimes\rho_2}$, where the \emph{message} $p=U^{\dagger}\rho_1$ has norm at most $1$. Player $1$'s strategy therefore matters only through a vector in dimension $r\le\eta^{-2}$.

This gives an algorithm by discretizing the message. For a net radius $\delta>0$, enumerate a $\delta$-net of the unit ball of $\Reals^r$. For each net point $g$, use convex optimization to find $\rho_1\in\cK_1$ with $\|U^{\dagger}\rho_1-g\|\le\delta$, if one exists. Then maximize $\ip{M}{Ug\otimes\rho_2}$ over $\rho_2\in\cK_2$ using the linear oracle, and return the best feasible pair. The compressed payoff is $1$-Lipschitz in the message, so the resulting payoff is at least $\OPT(M)-O(\eta+\delta)$. Taking $\delta=\eta$ gives running time $\poly(d)\cdot\exp(O(\eta^{-2}\log(1/\eta)))$.

For general $k$, applying bipartite compression independently at each site incurs an error of $O(k\eta)$, and requires guessing $k$ messages of dimension at most $\eta^{-2}$ each. Achieving accuracy $\epsilon$ requires taking $\eta=\Theta(\epsilon/k)$ leading to an overhead of $\exp(\tilde O_{\epsilon}(k^3))$. This bound is polynomial in $d^k$ only when $k$ is at most polylogarithmic in $d$. Two additional ideas overcome these losses.

\paragraph{Prefix Compression.} The first idea is to avoid guessing the compressed strategies of all players independently. Indeed, once the first $j$ players have chosen $\rho_1,\ldots,\rho_j$, the remaining players only need to know how the product $\rho_1\otimes\cdots\otimes\rho_j$ interacts with the remaining systems. We therefore compress the entire \emph{prefix} of strategies into a single low-dimensional message, which is updated as each new player chooses its strategy. Starting with $V_0=\Reals$, at step $j=1,\ldots,k-1$ we apply spectral truncation across the cut between $V_{j-1}\otimes\cH_j$ and the remaining systems. Let $V_j$ be the space of retained coordinates and $U_j: V_j\to V_{j-1}\otimes\cH_j$ their isometric embedding. The message representing $\rho_1\otimes\cdots\otimes\rho_j$ is defined recursively by
\[
p_0=1,\qquad
p_j=U_j^{\dagger}(p_{j-1}\otimes\rho_j)\in V_j,
\qquad \dim V_j\le\eta^{-2}.
\]

\paragraph{The Energy Argument.}
The second idea improves the accumulated error from $O(k\eta)$ to $O(\eta\sqrt k)$. At each step $j$, the spectral truncation splits the incoming message $p_{j-1}\otimes \rho_j$ into a retained component $p_j$ and a discarded component $q_j$. Since these components are orthogonal, and $\|\rho_j\|\le 1$, the squared norm of the discarded component is bounded by the decrease in the squared norm of the message, as $\|q_j\|^2\le \|p_{j-1}\|^2-\|p_j\|^2$.
Crucially, this simple observation lets us bound by a telescopic argument the total error by
\[
\textstyle
\eta\sum_{j=1}^{k-1}\|q_j\| \le \eta \sqrt k.
\]
To obtain an algorithm, we discretize each message space to accuracy $\delta$. Since each space has dimension at most $\eta^{-2}$, and we enumerate one space at a time, the resulting guarantees are a total error of $O(\eta\sqrt k+k\delta)$ and a running time of $\poly(d^k)\cdot \exp(O(\eta^{-2}\log (1/\delta)))$. This gives a polynomial-time algorithm whenever $k$ is polynomially bounded in $d$, by choosing $\eta=\Theta(\epsilon/\sqrt k)$ and $\delta=\Theta(\epsilon/k)$. For larger $k$, grouping players into blocks and using exhaustive search within each block gives the same bound.

\paragraph{Related Techniques.}
The idea of approximating solutions to such product objectives was already present in related domains, usually when the strategy space is a randomized choice over finitely many options and thus a simplex. This has yielded many quasi-polynomial-time algorithms in various settings. For instance, \citet{lipton2003playing,barman2018approximating,deligkas2022approximating} showed that mixed strategies in large games can be replaced by logarithmic-support empirical strategies. A closely related idea underlies the algorithmic results of \citet*{aaronson2014multiple}, where they devised a technique called \emph{peeling} to compute the value of a free-game (multi-prover interactive games in which each prover gets independent questions) by removing one prover at a time by guessing their strategy. Similar arguments were used for dense CSPs by \citet*{arora1995polynomial,manurangsi2015approximating}. 
\citet*{brandao2015estimating} use this idea to compute operator norms, best-separable state, and injective tensor norms. Recently, \citet*{bernasconi2026unifying} showed a unifying framework that captures some of these results, through the lens of fixed-degree polynomial optimization. The crucial distinction is that most of these algorithms require quasi-polynomial time (and necessarily so, given the lower bounds of \citet{aaronson2014multiple,rubinstein2017settling}). This stems from the different techniques used to concisely describe a guessed solution. In most of these results, the optimization space is a randomized strategy over finite options, and the key compression statement is that an empirical sample of such a distribution approximates the objective. By Hoeffding's inequality, this guarantees logarithmic support; thus, exhaustive enumeration leads to quasi-polynomial time. Our compression mechanism is fundamentally different. We compress a player's interaction with the remaining players via spectral truncation, showing that only a constant number of directions need to be preserved, which explains the quasi-polynomial vs.~polynomial improvement.

\subsection{Technical Outline II: Quantum Algorithms}\label{sec:ov-quantum}
In the property-testing setting, the algorithm receives copies of an unknown state $\rho\in\cD((\C^d)^{\otimes k})$ and wants to test separability with Frobenius gap $\eta>0$. A trivial solution for this problem would be to learn a classical description of $\rho$ via, e.g., tomography, and then run the algorithm described above. However, this would necessarily impose a runtime polynomial in $d^k$. To achieve runtime polynomial in $k\log d$, the challenge is to apply our ideas \emph{without learning an explicit description of} $\rho$. As a warm-up, let us first describe a tester whose runtime is polynomial in $k$ and $d$. At a later point in this overview, we will reduce the dependency to $\log d$ via a dimensionality reduction.

\paragraph{Building a Small Cover.} Starting from a fixed separable state, say $\sigma = \ket{0}\bra{0}^{\otimes k}$, our algorithm again proceeds by iterating Frank-Wolfe to update $\sigma$ in the direction of $\rho$, while maintaining $\sigma$ as a short mixture of product states, each stored as $k$ individual $d$-dimensional density matrices, so that the algorithm's runtime stays polynomial in $dk$. Each iteration of Frank-Wolfe requires solving the optimization problem
\begin{equation}\label{eq:ov-signed}
\max\Big\{\Tr[(\rho-\sigma)\ket{x}\bra{x}]:x=x_1\otimes\cdots\otimes x_k,\, x_j\in\C^d,\, \|x_j\|=1 \Big\}.
\end{equation}
The first difficulty is therefore implementing the optimization oracle for \Cref{eq:ov-signed}, given $\sigma$ explicitly but from copies of $\rho$.

Our starting point is the recent work of~\citet{bakshi2025learning}, which considers a closely related optimization problem to \Cref{eq:ov-signed}, where the objective function is $\Tr[ \rho\ket{x}\bra{x}]$ and the local dimension is fixed to $d=2$. At a high level, their algorithm builds a small cover of product vectors $a$ such that every $x$ (with non-negligible overlap with $\rho$) has a large overlap with at least one $a$. As in \cref{sec:ov-classical}, spectral compression allows to restrict the cover to a subspace of dimension $(kd)^{O(1)}$, so we can solve the optimization problem efficiently. 
To extend this to the signed objective function in \Cref{eq:ov-signed}, we observe that $\sigma\succeq0$, so any product vector with objective value at least $\theta$ also satisfies $\Tr[\rho\ket{x}\bra{x}]\ge\theta$. Therefore, the same cover remains valid for our objective function. Extending the algorithm from~\citet{bakshi2025learning} to an arbitrary local dimension $d$ requires additional technical care but follows the same overall approach.

\paragraph{Filtering by Frobenius Mass.} 
The next goal is to reduce the dimension from $d$ to $q=O(k^2/\eta^4)$ while preserving separability. The idea is to keep only the local directions that contribute significantly to the Frobenius norm of $\rho$. For each local site $j$, define $A_j=\Tr_{[k]\setminus \{j\}}(\rho^2)$. The eigenvalues of $A_j$ measure the contribution of each local direction to $\|\rho\|^2$. Since $\Tr(A_j)\le 1$, there can be at most $q$ eigenvalues larger than $1/q$. We therefore define $F_j=f(A_j)$, where $f$ is a spectral cutoff that is zero below $1/q$, one above $2/q$, and interpolates linearly in between.
The key observation is that for separable states, the filtering discards very little Frobenius mass. We can estimate this loss from copies of $\rho$: if it is too large, we safely reject (i.e., declare the state as entangled); otherwise, the state after the filter remains $O(\eta)$-close to $\rho$ in Frobenius norm. Since the filtering acts locally and preserves separability, the separability promise is preserved up to a constant loss in accuracy. 

Moreover, we can implement the filtering channel directly from copies of $\rho$ without explicitly learning $A_j$. The key observation is that a swap-test procedure on two copies of $\rho$, followed by tracing out all but the $j$ tensor components, yields 
\[
\frac12\sum_{b\in\{0,1\}} \ket b\bra b \otimes \Tr_{[k]\setminus\{j\}}(\rho+(-1)^b\rho^2).
\]
and the difference between the two (subnormalized) blocks is precisely equal to $\Tr_{[k]\setminus\{j\}}(\rho^2) = A_{j}$. This allows us to implement the evolution $e^{itA_j}$ via sample-based Hamiltonian simulation \citep{kimmel2017hamiltonian}.
On an eigenvector of $A_j$ with eigenvalue $\lambda$, this evolution simply introduces a phase $e^{it\lambda}$. Using phase estimation, we can store an approximation of \(\lambda\) in an auxiliary register and use it to perform the controlled rotation
\[
\ket{\lambda}\ket{0} \to \ket{\lambda}\left(f(\lambda)\ket{0}+\sqrt{1-f(\lambda)^2}\ket{1}\right).
\]
Measuring the auxiliary qubit, we keep the filtered state on outcome $0$ and replace it with a flag $\ket\bot$ on outcome $1$, thus implementing the desired filtering channel.

Overall, we reduced the \emph{effective} local dimension to $q=O(k^2/\eta^4)$ without losing the promise. However, these low-dimensional subspaces are unknown, and identifying their bases would require time polynomial in $d$. Thus, while filtering reduces the effective dimension of the state, it does not yet give an efficient way to run our tester in dimension $q$.

\paragraph{Extracting the Low-Dimensional Subspace.} Our final challenge is to \emph{isolate} the relevant low-dimensional subspace. Although each local support has dimension at most $q$, identifying a basis for it would require time polynomial in $d$. We therefore seek a way to run our low-dimensional tester without explicitly identifying these subspaces. 
Let $N$ be the number of copies needed by the low-dimensional tester.
Analytically, we can rewrite our state as $\omega = (\bigotimes_j V_j) \hat{\rho} (\bigotimes_j V_j^\dagger)$ where $\hat{\rho}$ has local dimension $q$ and $V_j:\mathbb C^q\to\mathbb C^{d}$ are unknown isometries.

The key observation is that we do not need to recover $\hat{\rho}$ exactly: local unitary rotations preserve both separability and Frobenius distance. This motivates using Schur--Weyl duality, which separates the information about the unknown local basis from the information invariant under local rotations.
More precisely, the Schur transform \citep{harrow_2005} decomposes $N$ copies of a local system into registers $Q_\lambda$ and $P_\lambda$, indexed by partitions $\lambda$ of $N$ (with at most $q$ rows in its Young diagram representation).
In this basis, the unknown embedding takes the form
\[
\mathsf{Sch}_{d}\,V_j^{\otimes N}\,\mathsf{Sch}_q^\dagger = \bigoplus_{\substack{\lambda\vdash N}} V_{j,\lambda}\otimes I_{P_\lambda}.
\]
Thus, the unknown isometries $V_{j,\lambda}:Q_\lambda^{(q)} \to Q_\lambda^{(d)}$ act only on the $Q_\lambda$ registers, leaving the $P_\lambda$ registers unchanged.
We can therefore apply the Schur transform, discard each $Q_\lambda^{(d)}$ register, and replace it with a fixed state on $Q_\lambda^{(q)}$, while preserving the sector labels and permutation registers, including their correlations across sites. 
Although this replacement loses information, the crucial observation is that it preserves the state after averaging over local unitary rotations. Indeed, Haar averaging removes coherences between different Schur sectors and replaces each $Q_\lambda$ register by the maximally mixed state. The averaged state therefore depends only on the information preserved by our procedure.

Let $\sigma$ denote the output of this algorithm on $\omega^{\otimes N}$.
For independent Haar-random unitaries $U_j\in U(q)$, write
$U=\bigotimes_{j=1}^k U_j$. Then
\[
\mathbb{E}_{U_1,\ldots,U_k}
\left[U^{\otimes N}\sigma(U^\dagger)^{\otimes N}\right]
=
\mathbb{E}_{U_1,\ldots,U_k}
\left[(U\hat{\rho}U^\dagger)^{\otimes N}\right].
\]

Therefore, independently at each site, we sample a Haar-random $U_j$
and apply \emph{that same unitary} to all $N$ registers at site $j$.
The averaged output is a mixture of $N$-fold copies of local-unitary
conjugates of $\hat{\rho}$. The randomized state satisfies the same
separability promise as $\hat{\rho}$, because local unitaries preserve
separability and Frobenius distance. Hence our low-dimensional tester applies.

\subsection{AI Disclosure}

The first part of the paper (classical algorithms) was developed without the use of AI. In the second part of the paper (quantum algorithms), we used ChatGPT to extend the algorithm of \citet{bakshi2025learning} to arbitrary local dimensions and to navigate representation theory for the dimensionality reduction. The entire paper was written by us, and we take full responsibility for its content.

\section{Preliminaries}

All Hilbert spaces are finite-dimensional. We write
$[k]=\{1,\ldots,k\}$, $B_{\cH}=\{x\in\cH:\|x\|\le1\}$,
and $\cH_{i:j}=\cH_i\otimes\cdots\otimes\cH_j$.
Tensor products carry the inner product
\[
\left\langle\otimes_j x_j,\otimes_j y_j\right\rangle
=\prod_j\langle x_j,y_j\rangle,
\qquad
\left\|\otimes_j x_j\right\|=\prod_j\|x_j\|.
\]
We use $\|\cdot\|$ for the induced Hilbert-space
norm, also called the Frobenius norm for matrices and tensors.
We write $A^{\dagger}$ for the adjoint of $A$. For matrices,
\[
\|A\|=\sqrt{\Tr(A^\dagger A)},\qquad
\|A\|_{\mathrm{op}}=\sup_{\|x\|=1}\|Ax\|,\qquad
\|A\|_1=\Tr\sqrt{A^\dagger A}.
\]
The Hermitian operators $\Herm(\cH)$ form a real Hilbert space
with inner product $\langle A,B\rangle=\Tr(AB)$.
The density matrices on $\cH$ are
\[
\cD(\cH)=\{\rho\in\Herm(\cH):\rho\succeq0,\ \Tr(\rho)=1\},
\qquad \cD_d=\cD(\C^d).
\]
Every $\rho\in\cD(\cH)$ satisfies $\|\rho\|\le1$.
A pure state is $\ket{x}\bra{x}$ with $\|x\|=1$.
The set of separable states is
\[
\Sep(d,k)=\conv\left\{
\bigotimes_{j\in[k]}\ket{x_j}\bra{x_j}:
x_j\in\C^d,\ \|x_j\|=1
\right\}.
\]
For $S\subseteq[k]$, $\Tr_S$ denotes the partial trace over
the systems in $S$; in particular,
$\Tr_{-j}=\Tr_{[k]\setminus\{j\}}$ and
$\rho^{[j]}=\Tr_{\{j+1,\ldots,k\}}(\rho)$.
Access to $N$ copies of $\rho$ means access to $\rho^{\otimes N}$.
For a nonempty set $K$ in a Hilbert space, define $\dist(x,K)=\inf_{y\in K}\|x-y\|$.
Weak membership with gap $\epsilon>0$ asks to distinguish
$x\in K$ from $\dist(x,K)\ge\epsilon$, with no requirement
on inputs between these cases.
A finite set $\cG\subseteq K$ is a $\delta$-net if every
$x\in K$ satisfies $\|x-g\|\le\delta$ for some $g\in\cG$.
For $\delta>0$, the unit ball of $\Reals^r$ admits such a net
of size at most $(1+2/\delta)^r$.

\part{Classical Algorithms}\label{part:classical}
\section{Spectral Compression}

The main theorem of this section is the following compression theorem for $k$-partite tensors.

\begin{theorem}\label{th:maincompression}
    Given finite-dimensional real Hilbert spaces $\cH_1,\ldots, \cH_k$ and $T\in \cH_1\otimes\cdots\otimes \cH_k$, with $\|T\|\le 1$ and any $\eta>0$, one can compute in polynomial time real vector spaces $(V_j)_{j\in[k-1]}$, with $\dim V_j\le \eta^{-2}$, linear maps $U_j:V_j\to V_{j-1}\otimes \cH_j$ for all $j\in[k-1]$ and a terminal tensor $\wt T\in V_{k-1}\otimes \cH_k$. Moreover, for any $\rho_1,\ldots, \rho_k$, with $\rho_j\in\cH_j$ and $\|\rho_j\|\le 1$, define recursively:
    \[
    p_0=1,\quad p_j = p_j(\rho_1,\ldots, \rho_{j})=U_j^{\dagger}(p_{j-1}\otimes \rho_j)\in V_j
    \]%
    then we have
    \[
    |\ip{T}{\rho_1\otimes\cdots\otimes \rho_k}-\ip{\wt T}{p_{k-1}\otimes \rho_k}|\le \eta\sqrt k.
    \]
\end{theorem}

The theorem says that any product vector $\rho_1\otimes\cdots\otimes\rho_k$ can be iteratively ``summarized'' by maintaining only a small amount of information about each prefix $\rho_1\otimes\cdots\otimes\rho_j$, and if the tensor $T$ is compressed to $\wt T$, then the value is preserved up to error $\eta\sqrt{k}$. The dependence of $O(\eta\sqrt{k})$ on the number of spaces $k$ is crucial to obtain polynomial-time algorithms, and one of the main technical challenges is to reach such a dependence, instead of the more straightforward $O(\eta k)$ theorem. In \Cref{sec:kpartite} we explain in more detail why this quantitative improvement is needed. In \Cref{sec:bipartite} we state a spectral compression lemma for bipartite systems. The idea is to keep only the singular directions corresponding to large singular values. \citet{bakshi2025learning} recently used similar ideas. 
In \Cref{sec:kpartite}, we present the proof of \Cref{th:maincompression}.

\subsection{Bipartite Spectral Compression}\label{sec:bipartite}

Intuitively, the lemma presented in this section states that given a tensor $S\in X\otimes Y$ of two real finite-dimensional Hilbert spaces $X$ and $Y$, we can find a suitable linear map $U^{\dagger}:X\to V$ to a low-dimensional real vector space $V$ with $\dim V\le \eta^{-2}\|S\|^2_{X\otimes Y}$, so that for any product test $x\otimes y$ for $x\in X$ and $y\in Y$, we cannot distinguish (up to error $\eta\|x\|\|y\|$) $\ip{S}{x\otimes y}$ to $\ip{\wt S}{U^{\dagger}x\otimes y}:=\ip{S}{UU^{\dagger}x\otimes y }$.

\begin{lemma}[Spectral Bipartite Compression]\label{lem:onecutcompression}
    Let $X$ and $Y$ be finite-dimensional real Hilbert spaces and let $S\in X\otimes Y$. For any $\eta>0$, we can find in polynomial time a Hilbert space $V=\Reals^r$, and a linear map $U:V\to X$ and tensors $\wt S=(U^{\dagger}\otimes I_Y)S\in V\otimes Y$ and $E=((I_X-UU^{\dagger})\otimes I_Y)S\in X\otimes Y$ such that:
    \begin{enumerate}
        \item $S=(U\otimes I_Y)\wt S+E$
        \item $r\le \eta^{-2}\|S\|_{X\otimes Y}^2$
        \item $|\ip{E}{x\otimes y}_{X\otimes Y}|\le \eta\cdot \|x\|_X \|y\|_Y$ for all $x\in X$ and $y\in Y$.
        \item $\|\wt S\|_{V\otimes Y}\le\|S\|_{X\otimes Y}$ and $U^{\dagger}U = I_V$
    \end{enumerate}
\end{lemma}

\begin{proof}%
    Given the tensor $S\in X\otimes Y$, we associate the linear map $A_S:X\to Y$, defined by $\langle S,x\otimes y\rangle_{X\otimes Y}=\langle A_Sx,y\rangle_Y$. Now we can compute the singular value decomposition of $A_S = \sum_{t\in[K]} \sigma_t R_t L_t^{\dagger}$, where $\sigma_1\ge\sigma_2\ge\cdots\ge \sigma_K\ge  0$ is a decreasing family of singular values, $K$ is the rank of $A_S$, $\{R_t\}_{t\in[K]}\subseteq Y$, $\{L_t\}_{t\in[K]}\subseteq X$ are families of orthonormal vectors. %

    Define $r$ as the number of singular values that exceed $\eta$. Observe that $r$ is bounded since
    \[
    r\eta^2\le \sum_{t\in [K]}\sigma_t^2=\|A_S\|^2=\|S\|^2_{X\otimes Y},
    \]
    which proves that $r\le \eta^{-2}\|S\|^2_{X\otimes Y}$. 

    Now we are ready to define the map $U:V\to X$, which we can define by the image of the basis of $V$ as $Ue_t=L_t$ for all $t\le r$.

    Then we have $S=(U\otimes I_Y)\wt S+E$, since 
    \begin{align*}
        (U\otimes I_Y)\wt S+E &= (U\otimes I_Y)(U^{\dagger}\otimes I_Y)S+((I_X-UU^{\dagger})\otimes I_Y)S\\
        &=(UU^{\dagger}\otimes I_Y)S+(I_X\otimes I_Y)S-(UU^{\dagger}\otimes I_Y)S\\
        &=(I_X\otimes I_Y)S = S
    \end{align*}

    Then we prove that the effect of $E$ is negligible. First observe that for all $x\in X$ and $t\le r$
    \begin{align*}
        L_t^{\dagger}(I_X-UU^{\dagger}) x = L_t^{\dagger} \left(x-\sum_{j\in[r]}L_jL_j^{\dagger} x\right)=L_t^{\dagger} x-\sum_{j\in [r]} L_t^{\dagger} L_j\cdot L_j^{\dagger} x=0,
    \end{align*}
    since $L_t^{\dagger} L_j =\mathbb{I}(t=j)$. Similarly, if $t>r$, the term $\sum_{j\in[r]}L_t^{\dagger}L_j$ is identical to $0$ and thus $L_t^{\dagger}(I_X-UU^{\dagger})x=L_t^{\dagger} x$.
    From this, and the singular value decomposition of $A_S$ we get
    \[
    A_S(I_X-UU^{\dagger})x = \sum_{t\in[K]}\sigma_t R_t L_t^{\dagger}(I_X-UU^{\dagger})x=\sum_{t>r}\sigma_t R_tL_t^{\dagger} x
    \]
    
    And thus: 
    \[
    \|A_S(I_X-UU^{\dagger})x\|_Y^2\le \sum_{t> r} \sigma_t^2\|R_t\|^2_Y|L_t^{\dagger} x|^2\le \eta^2\sum_{t\in[K]}|L_t^{\dagger} x|^2\le \eta^2\|x\|_X^2.
    \]
    This implies that
    \begin{align*}
    |\ip{E}{x\otimes y}_{X\otimes Y}| &= |\ip{((I_X-UU^{\dagger})\otimes I_Y)S}{x\otimes y}_{X\otimes Y}|\\
    &= |\ip{S}{(I_X-UU^{\dagger})x\otimes y}_{X\otimes Y}|\\
    &=|\ip{A_S(I_X-UU^{\dagger})x}{y}_Y|\\
    &\le \|A_S(I_X-UU^{\dagger})x\|_Y\cdot\|y\|_Y\\
    &\le \eta \|x\|_X\cdot \|y\|_Y.
    \end{align*}

    Finally, observe that
    \begin{align*}
        \|\wt S\|^2_{V\otimes Y} =\sum_{t\in[r]} \sigma_t^2\le \sum_{t\in[K]} \sigma_t^2 = \| S\|^2_{X\otimes Y},
    \end{align*} %
    concluding the proof.
\end{proof}

This proves that we can apply $UU^{\dagger}$ to any $x$ and preserve up to error $\eta(\|x\|_X\cdot \|y\|_Y)$ any product test:
\[
|\ip{S}{x\otimes y}-\ip{S}{UU^{\dagger}x\otimes y}|=|\ip{E}{x\otimes y}|\le \eta (\|x\|_X\cdot \|y\|_Y).
\]

This is useful algorithmically because $U^{\dagger}x$ lives in a low-dimensional space, and the second local system $y$ can now guess the strategy $x$ of the first local system $X$, since it enters the objective only through $p:=U^{\dagger}x\in \Reals^r$ (by allowing a constant error $\eta=O(1)$). Formally, if we assume that $\|S\|_{X\otimes Y}\le 1$ and $\|x\|_X\le 1$, then i) $p$ lies in a real vector space of constant dimension $r\le \eta^{-2} = O(1)$, and ii) the $\ell_2$ norm of $p$ is bounded, as 
\[
\|U^{\dagger}x\|^2_2 = \sum_{t\in[r]} |L_t^{\dagger} x|^2 \le \sum_{t\in[K]} |L_t^{\dagger} x|^2 \le \|x\|_X^2\le 1.
\]
Thus, all the possible strategies of the first player can be compressed in an $\ell_2$ ball in a space of constant dimension, and thus can be (approximately) enumerated in constant time. Once the guess has been made by the second player, we need to do two steps: i) check that the guess is realizable by some strategy $x$, and ii) optimize the second player's strategy against the guessed value.

This strategy already gives a deterministic constant-accuracy algorithm in the bipartite setting under the Frobenius norm. However, for a growing number of local systems, applying it independently at each local system leads to an error and running time that deteriorate too quickly with $k$. We therefore need the stronger recursive compression theorem.

\subsection{$k$-partite Spectral Compression and Proof of \Cref{th:maincompression}}\label{sec:kpartite}

\begin{algorithm}[!t]
\caption{Recursive $k$-partite Prefix Spectral Compression}\label{alg:cap2}
\begin{algorithmic}
\Require $T\in \cH_1\otimes\cdots\otimes\cH_k$ with $\|T\|\le1$ and $\eta>0$.
\State $V_0:=\Reals$, $T_0:=T$
\For{$j\in[k-1]$}
    \State $T_{j-1}\in (V_{j-1}\otimes \cH_j)\otimes \cH_{j+1:k}$
    \State Apply \Cref{lem:onecutcompression} with $X=V_{j-1}\otimes \cH_j$, $Y=\cH_{j+1:k}$ and $S=T_{j-1}$.
    \State Build accordingly $V_j=\Reals^{r_j}$, the linear map $U_j:V_j\to V_{j-1}\otimes \cH_j$ and the compressed tensor $T_{j}=(U_j^{\dagger}\otimes I)T_{j-1}$ %
\EndFor
\State \Return $((U_j,V_j)_{j\in[k-1]},T_{k-1})$
\end{algorithmic}
\end{algorithm}

We now formalize the recursive compression procedure outlined in \Cref{sec:ov-classical} and describe \Cref{alg:cap2} in more detail. Let $\cH_1,\ldots, \cH_k$ be finite-dimensional Hilbert spaces and 
\[
T\in \cH_1\otimes\cdots\otimes \cH_k,
\]
with $\|T\|\le 1$.
Formally, the idea behind \Cref{alg:cap2} is to start by using \Cref{lem:onecutcompression} across $\cH_1$ and $\cH_{2:k}$ producing the vector space $V_1$ and the embedding $U_1:V_1\to \cH_1$ and a compressed tensor $T_1=(U_1^{\dagger}\otimes I) T$\footnote{Here we use $\cH_{i:j} = \cH_i\otimes\cdots\otimes \cH_j$ for all $i<j\le k$.}. However, for the second player we do not simply compress across $\cH_2$ and $\otimes_{i\neq 2}\cH_i$, but we compound the compression by applying \Cref{lem:onecutcompression} across $V_1\otimes \cH_2$ and $\cH_{3:k}$, and in general, at step $j\in[k-1]$, we apply it across $V_{j-1}\otimes \cH_j$ and $\otimes_{i>j}\cH_i$. This means that we can maintain a single low-dimensional vector space $V_j$ in which we can summarize the entire interaction between the first $j$ systems and the remaining ones, instead of having independent $k$ low-dimensional local summaries of the strategies.

The final algorithm uses the last tensor $\wt T=T_{k-1}$ returned by \Cref{alg:cap2}, optimizes the linear operator over the last local system, and then reconstructs the entire solution by moving backward.

\begin{lemma}[Recursive $k$-partite Spectral Compression]\label{lem:kcutcompression}
    Let $\cH_1,\ldots, \cH_k$ be finite-dimensional Hilbert spaces and 
    \[
    T\in \cH_1\otimes\cdots\otimes \cH_k,
    \]
    with $\|T\|\le 1$. Then \Cref{alg:cap2} finds real spaces $(V_j)_{j\in[k-1]}$, linear maps $U_j:V_{j}\to V_{j-1}\otimes \cH_j$ for $j\in[k-1]$ and tensors $(T_{j})_{j\in[k-1]}$ (with $T_{j}\in V_{j}\otimes \cH_{j+1:k}$) such that (having defined $E_j:=((I-U_jU_j^{\dagger})\otimes I)T_{j-1}$ for all $j\in[k-1]$):
    \begin{enumerate}
        \item $T_{j-1}=(U_j\otimes I)T_j+E_j$
        \item $\|T_j\|\le 1$
        \item $\dim V_j \le \eta^{-2}$
        \item $|\ip{E_j}{u\otimes v}|\le \eta(\|u\|\cdot\|v\|)$ for all $u\in V_{j-1}\otimes \cH_j$ and $v\in \cH_{j+1:k}$.
    \end{enumerate}
\end{lemma}
\begin{proof}
    We apply \Cref{lem:onecutcompression} recursively. More precisely, we can observe that
\begin{align*}
\|T_j\| &=\|(U_j^{\dagger}\otimes I)T_{j-1}\|\\
&\leq \|U_j^{\dagger}\otimes I\|_{\mathrm{op}}\|T_{j-1}\|\\
&\leq\|T_{j-1}\|.
\end{align*}
Thus $\|T_j\|\leq1$ inductively, and consequently
\[
\dim V_j\leq\eta^{-2}\|T_{j-1}\|^2\leq\eta^{-2},
\]
showing that we can indeed apply \Cref{lem:onecutcompression} recursively over $j$.
\end{proof}

The algorithmic usefulness of this lemma is explained as follows. If the local strategies $\rho_1,\ldots, \rho_{j}$ have been fixed, we can define a recursive message to pass to the next systems as
\[
p_0=1,\quad\text{and}\quad p_j:=U_j^{\dagger}(p_{j-1}\otimes \rho_j)\in V_j.
\]

Now, just as before, $V_j$ remains a real low-dimensional space and the message $p_j$ has small $\ell_2$-norm:
\begin{align}\label{eq:iterativemex}
\|p_j\|^2=\|U_j^{\dagger}(p_{j-1}\otimes \rho_j)\|^2\le \|p_{j-1}\otimes \rho_j\|^2\le \|p_{j-1}\|^2.
\end{align}

We now prove the energy bound that controls the cumulative compression error, which allows us to collect a cumulative error of $O(\eta\sqrt k)$ instead of $O(\eta k)$.
To measure the information of the discarded message, we can compute the quantity $q_j=(I-U_jU_j^{\dagger})(p_{j-1}\otimes \rho_j)=p_{j-1}\otimes \rho_j-U_jp_j$.

\begin{lemma}[Energy Bound]\label{lem:energy}
    Given $\rho_j\in \cH_j$ with $\|\rho_j\|\le 1$ for all $j\in[k]$, compute $(p_1,\ldots, p_{k-1})$ according to \Cref{eq:iterativemex}. Then, if we let $q_j=(I-U_jU_j^{\dagger})(p_{j-1}\otimes \rho_j)$, we then have
    \[
    \sum_{j\in[k-1]}\|q_j\|^2\le 1.
    \]
\end{lemma}
\begin{proof}
    First observe that $q_j=(I-U_jU_j^{\dagger})(p_{j-1}\otimes \rho_j)=p_{j-1}\otimes \rho_j-U_jp_j$, and that $q_j$ is orthogonal to $U_jp_j$. Thus
    \[
    \|p_{j-1}\otimes \rho_j\|^2 = \|q_j\|^2+\|U_jp_j\|^2= \|q_j\|^2+\|p_j\|^2,
    \]
    where the last equality follows from the fact that $U_j^{\dagger}U_j=I$.
    Moreover, since $p_{j-1}\otimes \rho_j\in V_{j-1}\otimes \cH_j$, we have that 
    \begin{align*}
        \|p_{j-1}\otimes \rho_j\|^2&=\ip{p_{j-1}\otimes \rho_j}{p_{j-1}\otimes \rho_j}_{V_{j-1}\otimes \cH_j}\\
        &=\ip{p_{j-1}}{p_{j-1}}_{V_{j-1}}\ip{\rho_j}{\rho_j}_{\cH_{j}}\\
        &=\|p_{j-1}\|^2_{V_{j-1}}\|\rho_j\|^2_{\cH_{j}}\\
        &\le \|p_{j-1}\|^2.
    \end{align*}
    Combining the two, we obtain that 
    $\|q_j\|^2\le \|p_{j-1}\|^2-\|p_{j}\|^2$, from which it follows that
    \[
    \sum_{j=1}^{k-1}\|q_j\|^2\le \|p_0\|^2-\|p_{k-1}\|^2\le 1,
    \]
    since $p_0=1$.
\end{proof}

This energy-conservation property in compression, combined with \Cref{lem:kcutcompression}, gives the desired bound. For $j\in\{0\}\cup[k-1]$, the partially compressed value is
\[
    \Phi_j=\ip{T_{j}}{p_j\otimes\rho_{j+1}\otimes\cdots\otimes \rho_k},
\]
where, $(T_{j})_{j\in[k-1]}$ are the tensors produced by \Cref{alg:cap2}, and the $(p_{j})_{j\in[k-1]}$ are those produced by the iterative procedure of \Cref{eq:iterativemex} from any fixed $(\rho_1,\ldots,\rho_k)\in \cH_1\times\cdots\times \cH_k$. 

Rather than bounding the approximation error independently at each cut, we telescope the partially compressed values:
\[
\Phi_0-\Phi_{k-1} = \sum_{j\in[k-1]} (\Phi_{j-1}-\Phi_{j}),
\]
instead of each individual error.

\begin{proof}[Proof of \Cref{th:maincompression}]
    Recall from the notation of \Cref{lem:kcutcompression}, that $T_{j-1} = (U_j\otimes I)T_j+E_j$ and that $p_{j-1}\otimes \rho_j = U_jp_j+q_j$ from \Cref{lem:energy}. We also abbreviate $\rho_{i:j} = \rho_{i}\otimes\cdots\otimes\rho_j$. Thus:
    \begin{align*}
    \Phi_{j-1} &= \ip{(U_j\otimes I)T_j+E_j}{(U_jp_j+q_j)\otimes\rho_{j+1:k}}\\
    &= \ip{(U_j\otimes I)T_j}{U_jp_j\otimes \rho_{j+1:k}}+\ip{(U_j\otimes I)T_j}{q_j\otimes \rho_{j+1:k}}\\
    &\quad+\ip{E_j}{U_jp_j\otimes \rho_{j+1:k}}+\ip{E_j}{q_j\otimes \rho_{j+1:k}}\\
    &=\ip{T_j}{U_j^{\dagger}U_jp_j\otimes \rho_{j+1:k}}+\ip{T_j}{U_j^{\dagger}q_j\otimes \rho_{j+1:k}}\\
    &\quad+\ip{E_j}{U_jp_j\otimes \rho_{j+1:k}}+\ip{E_j}{q_j\otimes \rho_{j+1:k}}\\
    &=\ip{T_j}{p_j\otimes\rho_{j+1:k}}+\ip{E_j}{q_j\otimes \rho_{j+1:k}}\\
    &=\Phi_j+\ip{E_j}{q_j\otimes \rho_{j+1:k}},
    \end{align*}
    where $\ip{T_j}{U_j^{\dagger}q_j\otimes \rho_{j+1:k}}=0$ since $U_j^{\dagger}q_j=0$, and $\ip{E_j}{U_jp_j\otimes \rho_{j+1:k}}=0$ since $E_j=((I-U_jU_j^{\dagger})\otimes I)T_{j-1}$ and thus $(U_j^{\dagger}\otimes I)E_j=0$.
    Finally:
    \begin{align*}
    |\Phi_0-\Phi_{k-1}| &\le \sum_{j\in[k-1]} |\ip{E_j}{q_j\otimes \rho_{j+1:k}}|\\
    &\leq\eta\sum_{j\in[k-1]}\|q_j\|\cdot\|\rho_{j+1:k}\|\\
    &\le \eta\sum_{j\in[k-1]}\|q_j\| \\
    &\le \eta\sqrt{k}\left({\sum_{j\in[k-1]}\|q_j\|^2}\right)^{1/2}\le \eta\sqrt k,
    \end{align*}
    where the penultimate inequality is Cauchy-Schwarz and the last uses the bound of \Cref{lem:energy}. 
\end{proof}

\section{From Spectral Compression to Optimization}

The goal of this section is to design an algorithm for computing
\begin{align}\label{eq:OTP}
\OPT:= \max_{\rho_1\in\cK_1,\ldots,\rho_m\in\cK_m}\ip{M}{\rho_1\otimes\cdots\otimes \rho_m},
\end{align}
where $\cK_j\subseteq B_{\cH_j}$ are subsets of the $\ell_2$-ball of a finite-dimensional Hilbert space $\cH_j$ and $M\in\cH_1\otimes\cdots\otimes\cH_m$ is a tensor with Hilbert-Schmidt norm bounded by $1$. In particular, we focus on additive approximations. The algorithmic procedure exploits \Cref{th:maincompression}.
We make the content of this section modular in the sense that we do not provide a concrete algorithm here, but rather a meta-algorithm. 
The reason is that the construction of the reachable-message cover will depend on the structure of the sets $\cK_j$ and for our applications we will need different instantiations for $\cK_j$.

The key concept that turns \Cref{th:maincompression} into an algorithmic procedure is an efficient cover of the ``messages'' produced by prefix local systems. We view the low-dimensional spaces $V_j$ generated by \Cref{th:maincompression} as the spaces that contain all the essential information needed for the prefix local systems $\cH_{1:j}$ to be transmitted to the suffix local systems $\cH_{j+1:m}$. Then the algorithm needs to build a cover of such messages iteratively, while also storing witnesses for the local vectors $\rho_1,\ldots,\rho_j$ that could have produced each retained message $p_j\in V_j$.

In the following section, we give the precise definition of the cover and a meta-algorithm that, given a way to build such a cover, finds an additive approximation of $\OPT$ of \Cref{eq:OTP}.

\subsection{Reachable Messages}

Fix the maps $(U_j)_{j\in[m-1]}$ produced by \Cref{th:maincompression}.
We define all possible reachable messages recursively:
\[
\cR_0^{exact} := \{1\},\quad\text{and}\quad \mathcal R_j^{{exact}}:=\left\{U_j^{\dagger}(p\otimes\rho):p\in\mathcal R_{j-1}^{{exact}}, \rho\in\mathcal K_j\right\}.
\]

We now define the recursive notion of cover used by the algorithm.

\begin{definition}[Cover of reachable messages]\label{def:reachable}
Let $\cS_0,\ldots,\cS_{m-1}$ be sets such that $\cS_j\subseteq V_j$ for all $j\in\{0\}\cup [m-1]$.
Given any $\cS_{j-1}\subseteq V_{j-1}$, we define the set of messages reachable at layer $j$ as
\[
    \cR_j(\cS_{j-1}):=\left\{U_j^{\dagger}(p\otimes\rho): p\in\cS_{j-1},\ \rho\in\cK_j\right\}\subseteq V_j.
\]
A sequence $\cS_0,\ldots,\cS_{m-1}$ is a $\delta$-cover of the reachable messages if $\cS_0=\{1\}$ and, for every $j\in[m-1]$: $\cS_j\subseteq \cR_j(\cS_{j-1})$, and for every $q\in\cR_j(\cS_{j-1})$ there exists $\hat q\in\cS_j$ such that $\|q-\hat q\|\le\delta$.
For every $\hat q\in\cS_j$ we additionally store the witness $p\in\cS_{j-1}$ and
$\rho\in\cK_j$ such that
\[
    \hat q=U_j^{\dagger}(p\otimes\rho),
\]
so that the corresponding strategies can be recovered at the end.
\end{definition}

Then we show that this definition actually produces a cover of $\cR_j^{exact}$.

\begin{lemma}\label{claim:close}
    Let $\cS_0,\ldots,\cS_{m-1}$ be a $\delta$-cover of reachable messages.
    For all $p^\star_j\in\cR_j^{exact}$ there exists $\hat p_j\in \cS_j$ such that $\|\hat p_j-p_j^\star\|\le j\delta$.
\end{lemma}
\begin{proof}[Proof of \Cref{claim:close}]
    We prove this by induction. It clearly holds for $j=0$ since $\cR^{exact}_0=\{1\}$ and $\cS_0=\{1\}$. Now suppose it holds for $j-1$. By definition, any message $p_j^\star\in\cR_j^{exact}$ can be written as $p_j^\star=U_j^{\dagger}(p_{j-1}^\star\otimes \rho_j)$, and by induction hypothesis there exists a $\hat p_{j-1}\in \cS_{j-1}$ such that 
    \[
        \|\hat p_{j-1}-p_{j-1}^\star\|\le (j-1)\delta.
    \]
    Define $q:=U_j^{\dagger}(\hat p_{j-1}\otimes\rho_j)\in\cR_j(\cS_{j-1})$.
    Since $\cS_j$ is a $\delta$-cover of $\cR_j(\cS_{j-1})$, there exists
    $\hat p_j\in\cS_j$ such that
    \[
        \|\hat p_j-q\|\le\delta.
    \]
    Then
    \begin{align*}
        \|\hat p_j-p_j^\star\|\le \|\hat p_j-q\|+\|U_j^{\dagger}((\hat p_{j-1}-p_{j-1}^\star)\otimes\rho_j)\|\le \delta+\|\hat p_{j-1}-p_{j-1}^\star\|\|\rho_j\|\le j\delta,
    \end{align*}
   concluding the proof of the claim.
\end{proof}

\subsection{Meta Algorithm}

The following meta-algorithm leaves two subroutines unspecified: the construction of the reachable-message cover and the optimization of the terminal tensor.

\begin{algorithm}[!t]
\caption{Meta Algorithm for \Cref{eq:OTP}}\label{alg:meta}
\begin{algorithmic}[1]
\Require \\
\begin{enumerate}
    \item $M\in \cH_1\otimes\cdots\otimes\cH_m$ with $\|M\|\le 1$ and $\epsilon>0$;
    \item an algorithm to build a $\delta$-cover of reachable messages as per \Cref{def:reachable};
    \item an algorithm to compute $\max_{\rho_m\in \cK_m}\ip{\wt T}{x\otimes \rho_m}$ for all $x\in X$, and $\wt T\in X\otimes \cH_m$.
\end{enumerate}
\State Set $\eta=\Theta(\epsilon/\sqrt{m})$ and $\delta=\Theta(\epsilon/m)$.
\State Apply \Cref{alg:cap2} with approximation $\eta$ to obtain
\[
((V_j,U_j))_{j\in[m-1]}
\qquad\text{and}\qquad
\wt M\in V_{m-1}\otimes \cH_m.
\]
\State Get $\cS_0,\ldots, \cS_{m-1}$ from the algorithm building the reachable-message $\delta$-cover.
\State Solve $G(p)=\max_{\rho_m\in \cK_m}\ip{\wt M}{p\otimes \rho_m}$ for all $p\in \cS_{m-1}$ using the optimization algorithm.
\State Record the $p_{m-1}\in \cS_{m-1}$ that attains the largest value and the corresponding maximizer $\rho_m$.
\For{$j=m-1,\ldots,1$}\Comment{Extract local vectors}
    \State From $p_j$, using the witnesses given by the $\delta$-cover, find $\rho_{j}$ and $p_{j-1}$ %
\EndFor

\State \Return $\rho_1,\ldots,\rho_m$
\end{algorithmic}
\end{algorithm}

\begin{lemma}\label{th:meta}
    Let $\cK_j\subseteq B_{\cH_j}$ be compact for all $j\in[m]$.
    Instantiate \Cref{alg:meta} with:
    \begin{enumerate}
    \item an algorithm that constructs a $\delta$-cover
    $\cS_0,\ldots,\cS_{m-1}$ of the reachable messages, together with
    witnesses as in \Cref{def:reachable}, in time $T_{\mathrm{cover}}$;
    \item an algorithm that, given $p\in V_{m-1}$, computes
    \[
        \max_{\rho_m\in\cK_m}
        \ip{\wt M}{p\otimes\rho_m}
    \]
    together with a maximizing $\rho_m$, in time $T_{\mathrm{term}}$.
    \end{enumerate}

    Then \Cref{alg:meta} finds $(\rho_1,\ldots,\rho_m)\in \cK_1\times\cdots\times\cK_m$ such that 
    \[
        \ip{M}{\rho_1\otimes\cdots\otimes\rho_m}\ge \max_{\rho_1\in\cK_1,\ldots,\rho_m\in\cK_m}\ip{M}{\rho_1\otimes\cdots\otimes\rho_m}-O(\epsilon),
    \]
    whenever the tensor $M$ has Hilbert-Schmidt norm bounded by $1$, in time 
    \[
    \poly(|M|,m)+T_{\mathrm{cover}}+|\cS_{m-1}|T_{\mathrm{term}}.
    \]
\end{lemma}
\begin{proof}
    First observe that $p_j$ will always be in the unit $\ell_2$ ball of $V_j$. Indeed
    \begin{align*}
        \|p_j\|&=\|U_j^{\dagger}(p_{j-1}\otimes \rho_j)\|\\
        &\le\|p_{j-1}\otimes \rho_j\|\\
        &\leq\|p_{j-1}\|,
    \end{align*}
    and by recursion it proves that $\|p_j\|\le 1$, since $p_0=1$.

    For a final message $p\in V_{m-1}$ define $G(p)=\max_{\rho_m\in \cK_m} \ip{\wt M}{p\otimes \rho_m}$.
    It is easy to see that this function is $1$-Lipschitz in the Euclidean $\ell_2$-norm.
    
    \renewcommand\qedsymbol{}
    \begin{claim}\label{claim:Lip}
        The function $p\to G(p)$ is $1$-Lipschitz.
    \end{claim}
    \begin{mdframed}[leftmargin=0pt, rightline = false, topline = false, bottomline = false, linecolor = black!99]
    \begin{proof}[Proof of \Cref{claim:Lip}]
    Simply observe that
        \begin{align*}
            |G(p)-G(p')|&=\left|\max_{\rho_m\in \cK_m} \ip{\wt M}{p\otimes \rho_m}-\max_{\rho_m\in \cK_m} \ip{\wt M}{p'\otimes \rho_m}\right|\\
            &\le \max_{\rho_m\in \cK_m}  |\ip{\wt M}{(p-p')\otimes\rho_m}|\\
            &\le \|\wt M\|\cdot \|p-p'\|\max_{\rho_m\in \cK_m} \|\rho_m\|\\
            &\le \|p-p'\|,
        \end{align*}
        which concludes the proof of the claim.
    \end{proof}
    \end{mdframed}
    \renewcommand\qedsymbol{$\square$}

    Now let $p^*_{m-1}\in \cR^{exact}_{m-1}$ be a reachable message originating from an optimal solution $\rho_1^*,\ldots,\rho_m^*$ to our problem.
    By \Cref{claim:close}, there exists some $\hat p_{m-1}\in\cS_{m-1}$ such that $\|\hat p_{m-1}-p^*_{m-1}\|\le m\delta$ and by \Cref{claim:Lip} we get:
    \begin{align*}
        G(\hat p_{m-1})&\ge G(p^*_{m-1})-m\delta\\
        &\ge \OPT-m\delta-\eta\sqrt{m}.
    \end{align*}
    where in the last inequality we used \Cref{th:maincompression}.
    Then, from $\hat p_{m-1}$ we can recover all the witnesses $\rho_1,\ldots, \rho_{m-1}$ that originated the message $\hat p_{m-1}$ by definition of the cover in \Cref{def:reachable}. Finally using again \Cref{th:maincompression} we get that $\ip{M}{\rho_1\otimes\cdots\otimes\rho_m}\ge G(\hat p_{m-1})-\eta\sqrt{m}$ showing that
    \[
    \ip{M}{\rho_1\otimes\cdots\otimes\rho_m}\ge \OPT - m\delta-2\eta\sqrt{m},
    \]
    and choosing $\delta=\Theta(\epsilon/m)$ and $\eta=\Theta(\epsilon/\sqrt m)$, we get an $O(\epsilon)$ approximation to $\OPT$.
\end{proof}

It remains to instantiate the two subroutines in the relevant setting. We will use different constructions depending on whether the number $k$ of local systems is smaller or larger than the local dimension $d^2$.

\section{Product Optimization and Weak Membership}

In this section we use \Cref{th:meta} to prove the following theorem. In particular, the following proves that a linear oracle for the local systems is enough to build the objects required by \Cref{alg:meta}.

The proof is divided into two regimes: $k\le d^2$ and $k>d^2$. For the $k\le d^2$ case we can instantiate \Cref{alg:meta} directly and find an algorithm with final running time of $\poly(d^k)\cdot\exp\left(O_\epsilon \left(k\log k\right)\right)$, which is polynomial for every $k=\poly(d)$, since we only need $\log k=O(\log d)$ for this running time to be a PTAS.

However, this running time is not polynomial when $d\ll k$ (which is a fundamental regime; for instance, it models many entangled qubits). The previous theorem cannot address this case without removing the extra $\log k$ term in the exponent. Sadly, we could not find a cleverer analysis to remove that dependence directly, so we had to resort to a different analysis when the number of systems $k$ is significantly larger than the local dimension $d^2$.

In this regime, the idea is that when the local systems have a low-dimension $d$ (relative to $k$), we can do brute-force optimization over the product of a small number of local systems.

\begin{lemma}\label{th:mainvero} %
    Let $\cH_1,\ldots, \cH_k$ be finite-dimensional real Hilbert spaces with $d=\max_{j\in[k]} \dim \cH_j$ with $d\ge 2$. For each $j\in[k]$ let $\cK_j\subseteq B_{\cH_j}$ be a convex set and assume that given $v\in \cH_j$ and $\tau>0$ we can find $x\in\cK_j$ such that
    \[
    \ip{v}{x}\ge \max_{y\in \cK_j}\ip{v}{y}-\tau,
    \]
    in time polynomial in $d,\tau^{-1}$. Define $\cA:=\{x_1\otimes\cdots\otimes x_k:x_j\in \cK_j\}$, $\cC=\conv(\cA)$, then for every $\epsilon>0$, given $M\in\cH_1\otimes\cdots\otimes \cH_k$, with $\|M\|\le 1$, we can find $\rho_i\in \cK_i$, such that
    \[
    \ip{M}{\rho_1\otimes\cdots\otimes \rho_k}\ge \max_{\rho_1\in\cK_1,\ldots,\rho_k\in\cK_k}\ip{M}{\rho_1\otimes\cdots\otimes \rho_k}-\epsilon.
    \]
\end{lemma}

\begin{proof}
    We will divide the proof into two regimes.

    \paragraph{Regime I: $k\le d^2$.} Apply \Cref{alg:meta} directly. We thus have $((V_j,U_j))_{j\in[k-1]}$ and $\wt M\in V_{k-1}\otimes \cH_k$ with $r_j:=\dim V_j\le \eta^{-2}=\Theta(k/\epsilon^2)$. %

    For each $j\in[k-1]$ we consider $\cG_j$ which is a $\delta/4$-net in $\ell_2$-norm of the Euclidean ball of $V_j$, with $\delta=\Theta(\epsilon/k)$. It is standard to obtain an explicit $\cG_j$ such that there exists an absolute constant $C>0$, such that
    \[
    |\cG_j| = \left(\frac C \delta\right)^{r_j}\le \left(\frac C \delta\right)^{\eta^{-2}}=\exp\left(O\left(\frac{k}{\epsilon^2}\log\left(\frac{k}{\epsilon}\right)\right)\right).
    \]

    We now show how to construct the $\delta$-cover $\cS_0,\ldots,\cS_{k-1}$, required by \Cref{alg:meta}, inductively. Set $\cS_0=\{1\}$, and suppose $\cS_{j-1}$ has already been constructed. For each pair $(p,g)\in\cS_{j-1}\times\cG_j$, consider the convex optimization problem
    \[
    \min_{\rho\in\cK_j}\|U_j^{\dagger}(p\otimes\rho)-g\|,
    \]
    which can be solved with precision $\tau=\delta/8$ in time $\poly(d^k/\tau)$. If the value of the program is less than $\delta/2$ for the minimizer $\rho$, then we add $U^{\dagger}_j(p\otimes \rho)$ to $\cS_j$ and store the witnesses $p\in \cS_{j-1}$ and $\rho\in\cK_j$. As soon as we find a solution to such a program, for a particular $g\in\cG_j$ and $p\in\cS_{j-1}$, we stop looping over $p\in \cS_{j-1}$ and move to the next point in $\cG_j$. In this way, there is at most one point of $\cS_j$ for each $g\in \cG_j$ and $|\cS_j|\le|\cG_j|$. %

    We claim that the resulting $\cS_j$ is a $\delta$-cover according to \Cref{def:reachable}.
    We have to show that for all $p\in\cS_{j-1}$ and $\rho\in \cK_j$, and every reachable message $q(\rho,p):=U_j^{\dagger}(p\otimes \rho)\in\cR_j(\cS_{j-1})$ there exists $\hat q=U_j^{\dagger}(\hat p\otimes\hat \rho)\in \cS_j$ for some $\hat \rho\in\cK_j$ and $\hat p\in\cS_{j-1}$ with $\|q(\rho,p)-\hat q\|\le \delta$. 
    Since $\|q(\rho,p)\|\le 1$, we have that there exists $g\in\cG_j$ with $\|q(\rho,p)-g\|\le\delta/4$, which exists by definition of $\cG_j$. 
    Consider the convex problem with a tuple $(g,p)$. The problem is feasible since the original $\rho$ certifies it:
    \[
    \min_{\sigma\in \cK_j}\|U_j^{\dagger}(p\otimes\sigma)-g\|\le \|U_j^{\dagger}(p\otimes\rho)-g\|=\|q-g\|\le\delta/4.
    \]
    Therefore, since we are solving it with accuracy $\delta/8$ we are going to find some $\rho'$ such that
    \[
    \|U_j^{\dagger}(p\otimes \rho')-g\|\le \delta/4+\delta/8\le \delta/2.
    \]
    We are not guaranteed that such $p$ will be the one passing the test, but this proves that there exists some (possibly different) $p'\in \cS_{j-1}$ that does. In that case we will add some $\hat q=U_j^{\dagger}(p'\otimes \hat\rho)$ to $\cS_j$ such that $\|\hat q-g\|\le \delta/2$. Then 
    \[
    \|\hat q-q(\rho,p)\|\le \|\hat q-g\|+\|g-q(\rho,p)\|\le \frac\delta2+\frac\delta4\le\delta,
    \]
    which proves that we are building a $\delta$-cover.

    Moreover, it is clear that the time $T_{\mathrm{cover}}$ required to build the sets $\cS_0,\ldots,\cS_{k-1}$ is $\poly(d^k/\delta)\cdot O(k \left( C /\delta\right)^{2\eta^{-2}})$, since we have to repeat it $k$-times and each time it takes time $\poly(d^k /\delta)$ to solve the optimization problem, which we have to do at most $|\cG_j|\cdot|\cS_{j-1}|\le(\left( C /\delta\right)^{\eta^{-2}})^2$ times.
    
    We still need to implement the terminal optimization required by the meta-algorithm. For fixed $p\in V_{k-1}$, the problem %
    \[
        \max_{\rho_k\in\cK_k}
        \ip{\wt M}{p\otimes\rho_k}
    \]
    can be solved with a single call to the linear optimization oracle up to error $\tau=O(\epsilon)$. Thus the terminal problem is solvable in time $T_{\mathrm{term}}=\poly(d^k)$ since $\epsilon=O(1)$.
    
    Thus, according to \Cref{th:meta}, the total time required is
    \begin{align*}
        \poly(d^k)+T_\mathrm{cover} + |\cS_{k-1}| T_{\mathrm{term}} & \le \poly(d^k/\delta)\cdot O(k \left( C /\delta\right)^{2\eta^{-2}})+\left( C /\delta\right)^{\eta^{-2}} \poly(d^k)\\
        &\le \poly(d^k) \cdot \exp\left(O\left(\frac k{\epsilon^2}\log \frac k\epsilon\right)\right),
    \end{align*}
    where we plugged $\delta=\Theta(\epsilon/k)$ and $\eta=\Theta(\epsilon/\sqrt k)$.
    Now, if $k\le d^2$ and $\epsilon$ is fixed, then $\log(k/\epsilon)=O_{\epsilon}(\log d)$, and the running time is $d^{O_\epsilon(k)}$.

    \paragraph{Regime II: $k> d^2$.}

    Let $\gamma=\Theta(\epsilon/k)$. We construct an explicit $\gamma$-net $\cN_{\gamma,j}$ of $\cK_j$ with 
    \[
        |\cN_{\gamma,i}|
        \le
        \left(\frac{ C}{\gamma}\right)^{d},
    \]
    for some constant $C>0$. We claim that optimizing over $\cN_{\gamma,i}$ does not change the objective by a lot. Indeed, 
    For an optimal tuple $\rho_1,\ldots,\rho_k$ and for any $i$, choose $\rho'_i\in\cN_{\gamma,i}$ with $\|\rho_i-\rho'_i\|\le\gamma$. Then, it is easy to see that
    \begin{align*}
        \left\|\bigotimes_{i\in[k]}\rho_i-\bigotimes_{i\in[k]}\rho_i'\right\|\le \sum_{i\in[k]}\|\rho_i-\rho_i'\|\le k\gamma,
    \end{align*}
    by using inductively that $\|\rho_1\otimes\rho_2-\rho_1'\otimes\rho_2'\|\le \|\rho_1-\rho_1'\|+\|\rho_2-\rho_2'\|$ (see also \citet[Exercise~9.1.3]{wilde2013quantum}).
    
    Consequently, since $\|M\|\le 1$:
    \begin{equation}\label{eq:pure-discretization}
        \max_{\rho_i\in\cN_{\gamma,i}}
         \ip{M}{\rho_1\otimes\cdots\otimes\rho_k}
        \ge
        \max_{\rho_i\in\cK_i}
        \ip{M}{\rho_1\otimes\cdots\otimes\rho_k}-k\gamma.
    \end{equation} 

    Then partition $[k]$ into $m=\Theta(\sqrt{dk})$ blocks $I_1,\ldots,I_m$ and invoke \Cref{alg:meta} with 
    \[
        \cH'_j:=\bigotimes_{i\in I_j}\cH_i,
        \quad\text{and}\quad
        \cK'_j:=\left\{\bigotimes_{i\in I_j}\rho_i: \rho_i\in\cN_{\gamma,i}, \forall i\in I_j\right\},
    \]
    where the tensor $M'\in \cH'_1\otimes\cdots\otimes\cH'_m$, is the regrouping of $M\in \cH_1\otimes\cdots\otimes\cH_k$. Moreover, observe that $|\cK'_j| = |N_{\gamma,i}|^{O(k/m)}=(C\sqrt d/ \gamma)^{O(dk/m)}$.

    Under this instantiation, the conditions for \Cref{th:meta} hold and \Cref{alg:meta} computes $((V_j,U_j))_{j\in[m-1]}$ and $\wt M\in V_{m-1}\otimes \cH_m'$ with $r_j:=\dim V_j\le \eta^{-2}$.
    To build the $\delta$-cover $\cS_0,\ldots, \cS_{m-1}$, we apply the following simple brute-force procedure. First, define $\cG_j$, as the $\delta/4$-net of the radius $1$, $\ell_2$ ball of $V_j$, as in the case of $k\le d^2$. Then, assuming that we already built $\cS_{j-1}$, we enumerate every pair $(p,\rho)\in \cS_{j-1}\times \cK'_j$, computing $q=U_j^{\dagger}(p\otimes \rho)$, and then finding the closest $g\in \cG_j$ to $q$. If $\|q-g\|\le \delta/4$, and such $g$ was not already marked, we add $q$ to $\cS_j$ (moreover we store the witnesses $(p,\rho)\in \cS_{j-1}\times \cK_j$), and mark $g$. In this way we are going to add at most $|\cG_j|$ points to $\cS_j$, guaranteeing that $|\cS_j|\le (C/\delta)^{\eta^{-2}}$.

    This produces a $\delta$-cover according to \Cref{def:reachable}. Indeed, consider any $q(\rho,p)=U_j^{\dagger}(p\otimes \rho)\in \cR_j(\cS_{j-1})$ for any $\rho\in \cK_j'$ and $p\in \cS_{j-1}$. Take the point $g\in\cG_j$ chosen by the procedure for $(p,\rho)$. Then $\|q(\rho,p)-g\|\le\delta/4$ and $g$ has been marked, so there is a tuple $(p',\rho')\in \cS_{j-1}\times \cK'_j$, such that $\|g-q'\|\le \delta/4$, where $q'= U_j^{\dagger}(p'\otimes \rho')$ and $q'$ was added to $\cS_j$ (it might also be the case that $(\rho',p') = (\rho,p)$ but not necessarily). Then $\|q'-q\|\le \|q'-g\|+\|g-q\|\le \delta/2\le \delta$.
    Moreover, the time $T_\mathrm{cover}$ required to build the cover, is thus bounded by $\poly(d^k)\cdot (m \sqrt d /\gamma)^{O(dk/m)}(C/\delta)^{2\eta^{-2}}$, as we need to iterate for $m$ times over all elements of $\cK'_j$ and $\cS_{j-1}$, and search $\cG_j$ for each pair.
    
    For optimizing the terminal tensor $\wt M'$, we simply enumerate over the finite set $\cK'_{m}$, which takes time $T_{\mathrm{term}}\le |\cK_m'|\poly(|\wt M'|) = (C\sqrt d/\gamma)^{O(dk/m)}\poly(d^k)$.
    
    The total required time is thus, according to \Cref{th:meta}: %
    \begin{align*}
        \poly(d^k)+T_\mathrm{cover} + |\cS_{m-1}| T_{\mathrm{term}} & \le \poly(d^k)\cdot m (\sqrt d/\gamma)^{O(dk/m)} (1/\delta)^{O(\eta^{-2})}\\
        &\le \poly(d^k) \cdot\exp\left(O\left(\frac{dk}{m}\log \frac{dk}\epsilon+\frac m{\epsilon^2}\log \frac m\epsilon\right)\right),
    \end{align*}
    Taking $m=\Theta(\sqrt{dk})$ gives an upper bound of $\poly(d^k) \cdot\exp\left(O_\epsilon(\sqrt {dk}\log dk)\right)$.
    Now, for $d^2\le k$, we have that $\sqrt{dk}\log dk \le k^{3/4}\log dk=O(k\log d)$, and thus the running time becomes $d^{O_\epsilon(k)}$ as desired.
\end{proof}

\begin{theorem}\label{th:maintheoremclassical}
    Under the same assumptions of \Cref{th:mainvero}
for every $\epsilon>0$, given $M\in\cH_1\otimes\cdots\otimes \cH_k$, we can distinguish between 
    \[
    M\in\cC,\quad\text{and}\quad \dist(M,\cC)\ge \epsilon,
    \]
    in time $d^{O_\epsilon(k)}$.
\end{theorem}
\begin{proof}
    If $\|M\|\ge 1+\epsilon$, then $\dist(M,\cC)\ge \epsilon$, since $\cC\subseteq B_{\cH_1\otimes\cdots\otimes \cH_k}$ and thus
    \[
    \dist(M,\cC)=\inf_{M'\in \cC}\|M-M'\|\ge \|M\|-1\ge \epsilon.
    \]
    So, we can solve the case $\|M\|\ge 1+\epsilon$ efficiently by computing $\|M\|$ directly (which takes polynomial time, since $M$ is given explicitly).

    Thus we can assume that $\|M\|\le 1+\epsilon$ and consider $F:\cC\to\Reals$ defined as
    \[
    F(Y)=\frac12\|Y-M\|^2,
    \]
    whose optimum over $\cC$ is $F^*:=\frac12\dist(M,\cC)^2$. We use the Frank-Wolfe algorithm to turn the linear maximization oracle of \Cref{th:mainvero} into an optimization algorithm for $F$. Frank-Wolfe asks to solve $\arg\min_{Z\in \cC}\ip{Y_t-M}{Z}$ which, thanks to linearity, is equivalent to
    \[
    \arg\min_{\rho_1\in\cK_1,\ldots,\rho_k\in\cK_k}\ip{Y_t-M}{\otimes_{j\in[k]}\rho_j},
    \]
    which can be solved by \Cref{th:mainvero} at approximation $\delta$ in time $d^{O_\delta(k)}$. Since the diameter of $\cC$ is bounded and $F$ is $O(1)$-Lipschitz, then Frank-Wolfe in $T$ iterations yields $Y_T\in\cC$ with
    \[
    F(Y_T)\le F^*+O(T^{-1})+O(\delta).
    \] 
    If $M\in \cC$, then, $F^* =0$ and after $T=\Theta(1/\epsilon^2)$ iterations and $\delta=\Theta(\epsilon^2)$ we will find that $F(Y_T)\le O(\epsilon^2)$. On the other hand if $\dist(M,\cC)\ge \epsilon$, then $F^*=\Omega(\epsilon^2)$, and similarly we will find $F(Y_T)=\Omega(\epsilon^2)$. The two cases are separated, and the statement follows, since the running time of the algorithm is $T d^{O_\delta(k)}=d^{O_\epsilon(k)}$.
\end{proof}

\subsection{Applications to Nuclear Norm and Separability}

In this section, we prove two corollaries of \Cref{th:maintheoremclassical} concerning weak membership for the nuclear norm of high-order tensors and quantum separability in the Frobenius norm.

\subsubsection{Nuclear Norm} 

The (real) nuclear norm for a tensor $M\in (\Reals^d)^{\otimes k}$ is defined as
\[
    \|M\|_{*} = \inf\Bigg\{\sum_{i=1}^r|\lambda_i|:M = \sum_{i\in[r]}\lambda_i (x_{1,i}\otimes\cdots \otimes x_{k,i}),\|x_{j,i}\|=1, x_{j,i}\in \Reals^d,r\in \Naturals, \lambda_i\in \Reals\Bigg\},
\]
and call $K^*$ the set of all tensors with nuclear norm at most $1$.
The problem is thus to decide if $M\in K^*$ or $\dist(M,K^*)\ge \eta$.

\begin{theorem}
    Given a tensor $M\in(\Reals^d)^{\otimes k}$ and any $\eta>0$, there exists an algorithm to decide if
    \[
    M\in K^*,\quad\text{or}\quad \dist(M,K^*)\ge \eta,
    \]
    in time $d^{O_\eta(k)}$.
\end{theorem}
\begin{proof}
    We can use \Cref{th:maintheoremclassical} by using the fact that $K^*$ is the convex hull of unit-norm rank-one tensors \citep[Proposition~3.1]{friedland2018nuclear}, i.e.,
    \[
    K^*=\conv\{x_1\otimes\cdots\otimes x_k: x_j\in B_{\Reals^d}\}.
    \]

    Thus we can use \Cref{th:maintheoremclassical} with $\cK_j=B_{\Reals^d}$. The linear maximization oracle is simply the Frobenius norm computation since for all $v\in \Reals^d$ 
    \[
    \max_{y\in B_{\Reals^d}}\ip{v}{y} = \|v\|,
    \]
    and this provides all the ingredients needed by \Cref{th:maintheoremclassical}.
\end{proof}
We note that nuclear norm can also be defined for complex tensors \citep{friedland2018nuclear}. Our algorithm can be extended to complex Hilbert spaces to solve weak membership in the unit ball of complex tensors. We leave the precise description of this extension as future work.

\subsubsection{Quantum Separability} 
Another important application is the quantum separability problem. 
For $d,k\in\Naturals$, let
\[
\Sep(d,k):=\conv\left\{\rho_1\otimes\cdots\otimes\rho_k:\rho_1,\ldots,\rho_k\in\cD_d\right\},
\]
denote the set of fully separable $k$-partite states. 
\begin{theorem}\label{th:classicalquantumsep}
    Given an explicit $M\in \Herm((\C^d)^{\otimes k})$, we can distinguish between 
    \[
    M\in \Sep(d,k),\quad\text{or}\quad \dist(M,\Sep(d,k))\ge \eta,
    \]
    in time $d^{O_\eta(k)}$.
\end{theorem}
\begin{proof}
    We use \Cref{th:maintheoremclassical} with $\cH_j=\Herm(\C^d)$, as a real Hilbert space of dimension $d^2$ with Hilbert--Schmidt, and $\cK_j=\cD_d$. Then $\conv(\rho_1\otimes\cdots\otimes\rho_k:\rho_j\in\cD_d)=\Sep(d,k)$. Moreover, linear optimization over $\cK_j$ is immediate since it is a greatest eigenvalue problem. %

    Thus, \Cref{th:maintheoremclassical} guarantees a solution in time $(d^2)^{O_\eta(k)}=d^{O_\eta(k)}$.
\end{proof}

\part{Quantum Algorithms}

\section{Product-State Optimization}\label{sec:frombakshi}

In this section, we will prove the following theorem, which states that given copies of a $k$-partite state $\rho$ of local dimension $q$, and a classical description of a low-rank separable state $\sigma$, we can find a pure state with approximately optimal fidelity to $(\rho-\sigma)$ in time and copies $\poly(kq)$.

\begin{theorem}[Signed Product-State Optimization]\label{th:mainfromBak}
    Given copies of a $k$-partite state $\rho\in \cD((\C^q)^{\otimes k})$ of local dimension $q$ and a separable state $\sigma$ given by
    \[
    \sigma=\sum_{\ell\in[T]} w_\ell(\sigma_{\ell,1}\otimes\cdots\otimes\sigma_{\ell,k}),
    \]
    for some weights $w$ and $\sigma_{\ell,j}\in\cD(\C^q)$, and $\eta>0$ and $\delta>0$, we can find pure states $x^*_1,\ldots, x_k^*$ such that with probability at least $1-\delta$:
    \[
    \Tr\left[(\rho-\sigma)\bigotimes_{i\in[k]}\ket{x^*_i}\bra{x^*_i}\right]\ge\max_{x_1,\ldots,x_k\in \C^q:\|x_i\|=1}\Tr\left[(\rho-\sigma)\bigotimes_{i\in[k]}\ket{x_i}\bra{x_i} \right]-O(\eta^2), 
    \]
    with $\poly_\eta(kq T\log 1/\delta)$ time and copies of $\rho$.
\end{theorem}

This extends \citet{bakshi2025learning} (which focuses on $k$ qubits) to arbitrary local dimension $q$ and to the difference between states (one given by copies and one explicitly separable). This latter extension is needed because we will combine it with the Frank-Wolfe algorithm, which assumes a linear maximization oracle of this structure, to solve the weak membership problem in $\poly(qk)$ time and copies.

In \Cref{sec:dimension-reduction}, we improve this result to $\poly(k\log q)$ time and copies via a black-box reduction. 

To obtain \Cref{th:mainfromBak}, we can use the main structure of \citet{bakshi2025learning}. To summarize briefly their approach, for a state $\rho$, their algorithm is made of two components: 1) a subroutine that, given a set of product states, outputs another product state which is ``far'' from all product states in the set and has good fidelity with $\rho$ (this is \cite[Lemma~5.8]{bakshi2025learning}), and 2) a way of turning this algorithm into an algorithm that produces a small cover $C$ of pure states such that for all product states with good fidelity with $\rho$, there is a ``similar'' one in the cover (\cite[Algorithm~5.5]{bakshi2025learning}). We modify the details surrounding their construction; however, we maintain the overall structure of the algorithm. The main reason why we can use the same construction even for a difference of states $\rho-\sigma$ is that a cover for $\rho$ is automatically a cover for $\rho-\sigma$.

First we define what a $\theta$-overlap product cover is. We define distances directly in terms of overlap, rather than using tangent distance as in \citet{bakshi2025learning}. We require that every product $x$ that has good overlap with $\rho$ has a representative $c$ that is sufficiently aligned with $x$.

\begin{definition}[$\theta$-overlap Product Cover]\label{def:cover}
    Let $\rho\in\cD((\C^d)^{\otimes k})$ and let $\theta>0$. A finite set $C$ is a $\theta$-overlap product cover for $\rho$ if for every $x=\otimes_{j\in[k]} x_j$ with $x_j\in \C^d$ and $\|x_j\|=1$ such that $\Tr(\rho\ket{x}\bra{x})\ge \theta$, there exists $c\in C$ such that $|\ip{x}{c}|\ge \theta^2/64$.  %
\end{definition}

\subsection{$\LS$} Similarly to \citet{bakshi2025learning}, we use a local search algorithm to build the cover. We use this algorithm both to build the cover and to optimize the final objective once the cover is built.

\begin{definition}[$\LS_{\alpha,\epsilon}(\omega,K,a)$]\label{def:LS} %
    Given $m\in\Naturals$, $\alpha,\epsilon>0$, failure probability $\delta$, and $\omega\in \cD((\C^q)^{\otimes m})$ and $K=\sum_{\ell\in[T] }c_\ell( \otimes_{j\in[m]}K_{\ell,j})$ (with $K_{\ell,j}\in \C^{q\times q}$, Hermitian with $\|K_{\ell,j}\|\le 1$ and $c_\ell\in\Reals$) and a reference product vector $a=\otimes_{j\in[m-1]}a_j$ (with $a_j\in \C^q$ and $\|a_j\|=1$) $\LS$ returns $\hat x=\otimes_{j\in[m]}\hat x_j$ and $v\in \Reals$ such that with probability $1-\delta$ the following hold:
    \begin{enumerate}
    \item If $\max_{x\in R(a,\alpha)}\Tr((\omega-K)\ket{x}\bra{x} )\ge \epsilon$ then
    \[
    \Tr((\omega-K)\ket{\hat x}\bra{\hat x})\ge \max_{x\in R(a,\alpha)}\Tr((\omega-K)\ket{x}\bra{x} )-\epsilon,
    \]
    where %
    \(
    R(a,\alpha)=\{x=\otimes_{j\in[m]}x_j:|\ip{a}{x_{[m-1]}}|^2\ge \alpha, \|x_i\|=1\};
    \)
    \item and 
    \(
    |v-\Tr((\omega-K)\ket{\hat x}\bra{\hat x})|\le \epsilon.
    \)
    \end{enumerate}
\end{definition}

Then we show that there exists a quantum algorithm which implements it efficiently.

\begin{restatable}{lemma}{lemmaLS}\label{th:LS}
    For every $\delta>0$, there exists a quantum algorithm (\Cref{alg:local-search}) implementing $\LS_{\alpha,\epsilon}(\omega,K,a)$ according to \Cref{def:LS}, with failure probability at most $\delta$ using $(mq)^{O(\log(G/(\alpha\epsilon)))}\log(1/\delta)$ copies of $\omega$ and total running time $\poly_{\alpha,\epsilon,G}(m,q,T,\log(1/\delta))$, where $G:=1+\sum_{\ell=1}^T|c_\ell|$.
\end{restatable}

\begin{algorithm}[!t]
\caption{$\LS_{\alpha,\epsilon}(\omega,K,a)$}
\label{alg:local-search}
\begin{algorithmic}[1]
    \Require \\
    \begin{itemize}
        \item Copies of $\omega\in\cD((\C^q)^{\otimes m})$;
        \item An explicitly given $K=\sum_{\ell=1}^T c_\ell \bigotimes_{i\in[m]}K_{\ell,i}$;
        \item A reference product vector $a=\otimes_{j\in[m-1]}a_j$, with $\|a_j\|=1$;
        \item parameters $\alpha, \epsilon$.
    \end{itemize}
    \State $h\gets 4e\log(32 (1+\sum_{\ell=1}^T|c_\ell|) /(\alpha\epsilon))$
    \State Initialize a list $\cL=\{a\otimes b\}$ where $b\in \C^q$ is any unit vector. %
    \ForAll{$S\subseteq [m-1]$ with $|S|\le 2\log(1/\alpha)$}
        \State Let $P_S$ be the projection onto the subspace $\cV^h_S$
        \State Perform subspace tomography of $\omega$ onto $\cV^h_S$ and obtain $\widehat \omega$\Comment{Subspace Tomography}
        \State $H_S:=\widehat\omega-P_SKP_S$.
        \State For $(z,u)\in\cK_S$ define $F_S$ according to \Cref{eq:classicalOBJ}

        \State Approximately maximize $F_S$
        over $(z,u)\in\cK_S$, obtaining $(\hat z,\hat u)$ \Comment{Polynomial Optimization}

        \If{$F_S(\hat z,\hat u)\ge \epsilon/2$}
            \State Define the normalized product vector $ x$ by
            \[
                x_i:=\frac{a_i+\hat z_i}{\sqrt{1+\|\hat z_i\|_2^2}}\qquad (i\in [m-1]\setminus S),
            \quad
            \text{and}
            \quad
                x_i:=\frac{\hat u_i}{\|\hat u_i\|_2}\qquad (i\in S \cup\{m\}).
            \]
            \State $\cL\gets\cL\cup \{x\}$
        \EndIf
    \EndFor
    \State Estimate $\Tr((\omega-K)\ket{x}\bra{x})$ up to accuracy $\epsilon/4$ for every $x\in \cL$.
    \State \Return the $x$ with the largest estimated value together with its estimate over all $x\in\cL$.
\end{algorithmic}
\end{algorithm}

The algorithm follows exactly the same template proposed by \citet{bakshi2025learning} extended to $\C^q$. In particular, the idea is to parameterize the product space around the reference vector, truncate to a low-dimensional subspace, and then perform tomography only on that subspace to obtain a classical representation that can be handled with a classical polynomial optimization algorithm.

Indeed, we will assume to have access to an algorithm that performs subspace tomography.

\begin{restatable}[Subspace Tomography (essentially {\citep[Lemma~5.9]{bakshi2025learning}})]{lemma}{subspacetomography}
\label{lem:subspacetomography}
    Let $\omega\in\cD((\C^q)^{\otimes m})$ be a quantum state and $P$ a known projector on a subspace $V$ of dimension $D$ given explicitly by its orthonormal product basis, i.e.,
    \[
    P=\sum_{\ell=1}^D\ket{a_\ell}\bra{a_\ell}\quad\text{and}\quad V=\textnormal{im}(P).
    \]
    For any $\delta,\epsilon>0$ there exists an algorithm that with $O(\frac{D^2}{\epsilon^2}\log(1/\delta))$ copies of $\omega$ outputs a classical representation $\widehat \omega$ supported on $V$ such that with probability at least $1-\delta$ it holds that
    \(
    \|\widehat \omega-P\omega P\| \le \epsilon.
    \)
    Moreover, the algorithm runs in time $\poly\left(m,q,D,\epsilon^{-1},\log 1/\delta \right)$.
\end{restatable}

Moreover, we assume access to a classical polynomial optimization subroutine that lets us optimize tensor polynomials of fixed degree (multiplied by the exponential of the negative of a convex function) over convex sets.

\begin{restatable}[Polynomial Optimization (essentially {\citep[Section~6]{bakshi2025learning}})]{lemma}{lemmapolynomialopt}
\label{lem:classicopt}
    Let $E$ be a finite-dimensional real Hilbert space of dimension $n$ and let $\cK\subset B_E$ be compact and convex. Let 
    \[
    p(y)=\sum_{j=0}^D\ip{T_j}{y^{\otimes j}},
    \]
    and $\sum_{j=0}^D\|T_j\|\le 1$, and $\phi:\cK\to[0,\infty)$ be convex and continuous. For every $\epsilon>0$ there exists an algorithm which returns $\hat y\in\cK$ such that, if $\OPT=\max_{y\in \cK} e^{-\phi(y)}p(y)\ge \epsilon$ then
    \[
    e^{-\phi(\hat y)}p(\hat y)\ge\OPT- \epsilon.
    \]
    The algorithm runs in time $\poly_{D,\epsilon}(n)$.
\end{restatable}

More formally, in $\LS$ we are given a reference product vector $a$, a state $\omega$ through copies, and a classically described sum of product operators $K$, and we want to maximize 
\[
x\mapsto \Tr((\omega-K)\ket{x}\bra{x})
\]
over the set
\[
R(a,\alpha) = \{x=\otimes_{j\in[m]}x_j:|\ip{a}{x_{[m-1]}}|^2\ge \alpha, \|x_i\|=1\}.
\]
The crucial fact is that even if $x$ differs from $a$ in all coordinates, the deviation is bounded independently of $m$. Indeed, we can write any $x$ as the product of 
\begin{align}\label{eq:reparam}
x_i=\frac{a_i+z_i}{\sqrt{1+\|z_i\|^2}}, \quad\text{with}\quad z_i\in a_i^\bot,
\end{align}
and thus $|\ip{a}{x_{[m-1]}}|^2=\prod_{i<m}\frac{1}{1+\|z_i\|^2}$. Since $x\in R(a,\alpha)$, defining $S\subseteq[m-1]$ as the sites at which $\|z_i\|\ge1$ gives
\[
\sum_{i\in [m-1]\setminus S}\|z_i\|^2\le 2\sum_{i<m} \log(1+\|z_i\|^2)\le 2\log(1/\alpha).
\]
which in turn implies that only few $i$'s (i.e., at most $O(\log(1/\alpha))$) can have large (e.g., greater than $1$) orthogonal weight $\|z_i\|$.\footnote{Note that in the reparametrization of \Cref{eq:reparam} we are not guaranteed that $\|z_i\|\le 1$.}

This motivates the entire algorithmic template. Indeed, we can first guess a set $S$ of cardinality $O(\log(1/\alpha))$ and consider the convex set $\cK_S\subset (\C^q)^{m}$ of parameters
\[
\cK_S = \left\{(z,u):z_i\in a_i^\bot, \|z_i\|\le 1, \forall i\in [m-1]\setminus S, \hspace{-0.2cm}\sum_{i\in[m-1]\setminus S} \hspace{-0.3cm}\|z_i\|^2\le 2\log\left(\frac1\alpha\right), \|u_i\|\le 1, \forall i\in S\cup \{m\}\right\}.
\]

The property we actually need is that each $x\in R(a,\alpha)$ is represented by some $(z,u)\in \cK_S$ for the appropriate $S\subseteq[m-1]$.

Moreover, for a guess $S$, the bound $\sum_{i\in [m-1]\setminus S}\|z_i\|^2\le O(\log(1/\alpha))$ implies that most of its mass is in the subspace $\cV^h_S$ defined as the subspace in which there are at most $h$ sites significantly away from the reference direction $a_i$:
\[
\cV^h_S:=\bigoplus_{R\subseteq [m-1]\setminus S:|R|\le h}\left(\bigotimes_{i\in ([m-1]\setminus S)\setminus R}\textnormal{Span}(a_i)\otimes\bigotimes_{i\in R} a_i^\bot\otimes\bigotimes_{i\in S\cup \{m\}}\C^q\right),
\]
and let $P_S$ be the orthogonal projection over this subspace.

Indeed for an $x$ parametrized as in \Cref{eq:reparam} by parameters $z_i$, satisfying the bound $\sum_{i\in [m-1]\setminus S}\|z_i\|^2=O(\log(1/\alpha))$, the projection $P_S$ keeps only the terms with $|R|\le h$, and we can show that (for suitable choice of $h$):
\[
\|(I-P_S)x\|^2\le 2^{-h}\le O(\epsilon^2/\|\omega- K\|^2_{op}).
\]
This means that instead of performing tomography of $\omega$ on the entire space, we can perform subspace tomography only on $\cV_S^h$, since $\dim(\cV_S^h)=\poly(mq)$, when all approximations are $O(1)$.
We want to approximately maximize $x\mapsto\Tr(H_S\ket{P_Sx}\bra{P_Sx})$ by exploiting the representation \Cref{eq:reparam}. For a fixed $S$ and $(z,u)\in \cK_S$ we define the function
\begin{equation}\label{eq:classicalOBJ}
    F_S(z,u)=e^{-\Phi(z)} \Tr(H_S\ket{w_S(z,u)}\bra{w_S(z,u)}),
\end{equation}
where $w_S(z,u)=P_S[\bigotimes_{i\in[m-1]\setminus S}(a_i+z_i)\otimes\bigotimes_{i\in S\cup \{m\}}u_i]$ is the unnormalized vector, and $\Phi(z)=\sum_{i\in [m-1]\setminus S}\log(1+\|z_i\|^2)$ accounts for the normalization.

\begin{algorithm}[!t]
\caption{$\Cover(\rho;\theta)$}
\label{alg:cover}
\begin{algorithmic}[1]
    \Require Copies of $\rho\in\cD((\C^d)^{\otimes k})$, $\theta>0$ and $\LS$ oracle
    \State $C_0\gets\{1\}$
    \For{$j=1,\ldots, k$}
        \State If $C_{j-1}=\emptyset$ \Return $\emptyset$
        \State $C_j\gets\emptyset$
        \State $\textsf{Stop} \gets \textsf{false}$ and $t=1$
        \While{$t\le 8/\theta$ and $\textsf{Stop} = \textsf{false}$}\Comment{Inner Loop}
            \State $t\gets t+1$
            \State $\varphi_j\gets\lambda\sum_{c\in C_j}\ket{c}\bra{c}$ with $\lambda=(8/\theta)^{2}$
            \ForAll{$a\in C_{j-1}$} 
                \State %
                Run $\LS_{\alpha,\frac\theta{16}}(\rho^{[j]},\varphi_j,a)$ which returns value $v_a$ and solution $\hat x_a$ with $\alpha=(\theta/8)^4$
            \EndFor
            \If{$\max_{a\in C_{j-1}}v_a\ge \frac34\theta$}
                \State $C_j\gets C_j\cup\{\hat x_{\hat a}\}$ where $\hat a\in\arg\max_{a\in C_{j-1}}v_a$
            \Else
                \State $\textsf{Stop} \gets \textsf{true}$
            \EndIf
        \EndWhile
    \EndFor
    \Return $C_k$
\end{algorithmic}
\end{algorithm}
\subsection{Overlap Product Cover} For any $\rho\in\cD((\C^d)^{\otimes k})$, \Cref{alg:cover} will use as a subroutine $\LS$ to build a $\theta$-overlap product cover as per \Cref{def:cover}.

Using $\LS$ to build a $\theta$-overlap cover is natural. The algorithm will build a cover $C_j$ for every factor $j$ for the $j$-partite marginal $\rho^{[j]}:=\Tr_{\{j+1,\ldots,k\}}(\rho)$. It will naturally use the fact that $\LS$ also uses an explicitly given separable $K$ to regularize the search by turning the constraint of adding new elements to the cover with small overlap with existing ones into a Lagrangian term ($\varphi_j$ in \Cref{alg:cover}).%

\begin{restatable}{lemma}{coverquantum}\label{lem:coverquantum}
    Given $\theta>0$, and $\rho\in \cD((\C^d)^{\otimes k})$ supplied by copies, \Cref{alg:cover} returns a $\theta$-overlap cover for $\rho$ with probability at least $1-O(k\delta/\theta^2)$ using at most $O_\theta(\poly(kd)\log 1/\delta)$ copies and time.
\end{restatable}

\subsection{Proof of \Cref{th:mainfromBak}}

\begin{algorithm}[!t]
\caption{Product-State Optimization}
\label{alg:?}
\begin{algorithmic}[1]
    \Require Copies of $\rho\in\cD((\C^d)^{\otimes k})$ and $\sigma=\sum_{\ell\in[T]} w_\ell(\sigma_{\ell,1}\otimes\cdots\otimes\sigma_{\ell,k})$, $\delta,\eta>0$
    \State Compute a $\theta/2$-cover $C$ for $\rho$, with $\theta=\min\{\eta^2/8,1\}$
    \ForAll{$c=\otimes_{j\in[k]}c_j\in C$}
        \State $(\hat x_c,v_c)\gets\LS_{\alpha,\epsilon}(\rho,\sigma,c_{[k-1]})$ where $\alpha=(\theta/16)^4$ and $\epsilon = \theta/16$
    \EndFor
    \State Record $\hat x=\hat x_c$ associated with the largest $\hat v=v_c$
    \For{$t=1,\ldots,O(1/\theta \log 1/\delta )$}
        \State Draw $b$ uniformly from a product basis of $(\C^d)^{\otimes k}$.
        \State Estimate $v_b=\Tr((\rho-\sigma)\ket{b}\bra{b})$ using $N=O(\log(1/\delta)/\eta^4)$ copies of $\rho$. %
    \EndFor
    \Return $\arg\max_{b} v_b$ or $\hat x$ if $\hat v\ge \max_{b} v_b$
\end{algorithmic}
\end{algorithm}

In this section we provide a proof of \Cref{th:mainfromBak} that uses \Cref{alg:cover} to build a cover and $\LS$ oracle to optimize once the cover is obtained. The only delicate part is that $\LS$ works only when the optimum exceeds a threshold. To handle the case where the optimum is close to zero, we use a probabilistic argument to show that a random pure product state is not too negative.

\begin{proof}[Proof of \Cref{th:mainfromBak}]
    Let $\OPT=\max_{x_1,\ldots,x_k\in \C^q:\|x_i\|=1}\Tr\left[(\rho-\sigma)\bigotimes_{i\in[k]}\ket{x_i}\bra{x_i} \right]$. We are going to divide the proof into two cases. 

    \paragraph{Case I: $\OPT\ge \theta/2$.}
        Let $x^*=\otimes_{j\in[k]}x^*_j$ be a global optimizer. Since $\OPT\ge \theta/2$ we have that 
        \[
        \Tr(\rho \ket{x^*}\bra{x^*}) = \OPT+\Tr(\sigma \ket{x^*}\bra{x^*})\ge \frac\theta2.
        \]
        Therefore, the guarantees of the cover apply, and there exists a $c\in C$ such that $|\ip{c}{x^*}|^2\ge\alpha$, which implies that $|\ip{c_{[k-1]}}{x^*_{[k-1]}}|^2\ge\alpha$ and $x^*\in R(c_{[k-1]},\alpha)$. When $\LS$ is called on such $c_{[k-1]}$, we are guaranteed that it will output a $\hat x_c$ with estimated value $v_c$ at least $\OPT-2\epsilon$.

    \paragraph{Case II: $\OPT< \theta/2$.}
        In this regime a trivial application of Markov's inequality proves that a random basis vector gives a value which is approximately zero and thus $O(\theta)$ approximation of the optimum.
        
       Observe that $\Tr(\rho)=\Tr(\sigma)=1$. Let $b=i_1\otimes\cdots\otimes i_k$ be a uniform product basis, i.e., each $i_j$ is uniform and independent in $[d]$. Thus $\mathbb{E}_{b}[\ket{b}\bra{b}]=\frac{1}{d^k}\sum_{i_1,\ldots, i_k}\ket{i_1,\ldots, i_k}\bra{i_1,\ldots, i_k}=\frac{I}{d^k}$ and
        \[
        \mathbb{E}_b[\Tr((\rho-\sigma)\ket{b}\bra{b})]=\frac{1}{d^k}\Tr(\rho-\sigma)=0,
        \]
        and $\Tr((\rho-\sigma)\ket{b}\bra{b})\in[-1,1]$ a.s..
        Thus, by Markov's inequality, we obtain
        \[
        \mathbb{P}\left[\Tr((\rho-\sigma)\ket{b}\bra{b})\ge -\frac{\theta}{2}\right]\ge \frac{\theta}{2+\theta}\ge \frac{\theta}{4}.
        \]
        Therefore, after $O(1/\theta\log 1/\delta)$ draws of such $b$, we have that at least one has value at least $-\frac{\theta}{2}$ with probability at least $1-\delta$, which is an $O(\theta)=O(\eta^2)$ approximation of the optimum in this regime.
\end{proof}

\subsection{From Optimization to Membership}

A simple consequence of \Cref{th:mainfromBak} (similarly to what we did in \Cref{part:classical}) is that we can use the specific structure of that linear optimization oracle to solve a weak membership problem. This is why we assume the state $\sigma$ is given classically as a small-rank tensor, as it is the state updated by the Frank-Wolfe machinery.

\begin{theorem}\label{th:quantum_lowdimmembership}
    For every $\eta>0$ there is a quantum algorithm which, given copies of a state $\rho\in \cD((\C^q)^{\otimes k})$, decides if 
    \[
    \rho\in \Sep(q,k),\quad\text{or}\quad\dist(\rho,\Sep(q,k))\ge \eta,
    \]
    with probability of error at most $1/3$, using $(kq)^{O_\eta(1)}$ time and copies of $\rho$.
\end{theorem}
\begin{proof}
    We proceed exactly as in the proof of \Cref{th:classicalquantumsep}, applying Frank-Wolfe to $F(\sigma)=\frac12\|\rho-\sigma\|^2$ with $\sigma\in \Sep(q,k)$. At every iteration, the required linear optimization oracle is
    \[
    \max_{\pi\in\Sep(q,k)}\Tr((\rho-\sigma)\pi),
    \]
    which by linearity is attained at a pure product state and can therefore be approximated to additive accuracy $O(\eta^2)$ by \Cref{th:mainfromBak}. After $O(\eta^{-2})$ iterations, and estimating the final value of $F$ to additive accuracy $O(\eta^2)$, the two promised cases are separated. The claimed time and copy bounds follow from \Cref{th:mainfromBak}.
\end{proof}

\section{Dimensionality Reduction}\label{sec:dimension-reduction}

In this section, we prove the following reduction from $k$-partite Frobenius weak separability in arbitrary local dimension $d$ to local dimension $O(k^2)$, with the dimension $d$ only appearing in the time complexity, which is increased by a polynomial in $O(\log d)$, while all other measures are kept polynomial.

\begin{theorem}[Dimensionality Reduction]\label{th:dimension-reduction}
For all $\eta>0$, let $q=O(k^2/\eta^4)$ with $q \leq d$.
Suppose we have a quantum algorithm, with access to $N$ copies of $\sigma\in\cD((\C^q)^{\otimes k})$, that uses time $T$
to distinguish
\[
\sigma\in\Sep(q,k) \qquad\text{from}\qquad \dist(\sigma,\Sep(q,k))\ge\frac\eta2
\]
with probability of mistake at most $1/10$.
Then, given copies of $\rho\in\cD((\C^d)^{\otimes k})$, one can distinguish
\[
\rho\in\Sep(d,k) \qquad\text{from}\qquad \dist(\rho,\Sep(d,k))\ge\eta
\]
with error at most $1/3$, using $\poly(k,N,\eta^{-1})$ copies and time
$T+\poly(k,N,\eta^{-1},\log d)$.    
\end{theorem}

As a corollary of this black-box reduction, together with \Cref{th:quantum_lowdimmembership}, we obtain the following theorem on the sample complexity of weak membership.

\begin{theorem}\label{th:mainquantum}
    For every $\eta>0$, there is a quantum algorithm that, given copies of an unknown state $\rho\in \cD((\C^d)^{\otimes k})$, distinguishes between
    \[
    \rho\in\Sep(d,k)\quad\text{and}\quad\dist(\rho,\Sep(d,k))\ge \eta,
    \]
    with probability of error at most $1/3$, using $k^{O_\eta(1)}$ copies of $\rho$ and total time $(k\log d)^{O_\eta(1)}$.
\end{theorem}

\Cref{th:dimension-reduction} asks us to turn a $k$-partite state $\rho$ with local dimension $d$ into a state with local dimension $q=O_\eta(k^2)$ by paying only $\poly\log d$ time. Intuitively, \Cref{part:classical} shows that this is possible by retaining only the large singular values of the tensor, but doing this requires polynomial time in the dimension of the state. The idea of this section is to learn such a decomposition without ever recording a classical description. Much of the technical difficulty is in implementing this procedure with a quantum algorithm. However, it is instructive to momentarily disregard the implementation details and to focus on the high-level idea of the reduction. The complete algorithm can be found in \Cref{alg:dimreduction}.

For each factor $j\in[k]$, define the ``local Frobenius-mass'' $A_j=\Tr_{-j}(\rho^2)$ (here $\Tr_{-j}=\Tr_{[k]\setminus \{j\}}$). 
Note that $A_j\succeq 0$ and $\Tr(A_j)=\Tr(\rho^2)\le 1$.
Then, for any $a>0$, we define a ramp function which is zero below $1/a$, $1$ above $2/a$ and interpolates linearly in between:
\[
f_a(t)=\begin{cases} 0 &\text{if}\quad t\le 1/a\\
at-1 &\text{if}\quad t\in(1/a,2/a)\\
1 &\text{if}\quad t\ge 2/a.
\end{cases}
\]
In particular we will use $f=f_{\Theta(k^2/\eta^4)}$ (details in the proof of \Cref{th:dimension-reduction}). Let $F_j=f(A_j)$, which, by abuse of notation, means applying $f$ to the eigenvalues of $A_j$ (which are real and nonnegative since $A_j\succeq 0$). So if $A_j=\sum_{r}\lambda_r\ket{v_r}\bra{v_r}$ then $F_j=\sum_{r} f(\lambda_r)\ket{v_r}\bra{v_r}$.

The crucial identity is that measuring the Frobenius mass after applying a local Hermitian operator $B$ on the $j$-th party is directly linked to the local Frobenius mass $A_j$ (which also justifies our naming). Indeed,
\[
\|(B\otimes I_{-j})\rho\|^2=\Tr((B^2\otimes I_{-j})\rho^2) = \Tr(B^2\Tr_{-j}(\rho^2)) = \Tr(B^2A_j).
\]
This means that 
\begin{align}\label{eq:errordimred}
e_j=\Tr((I-F_j^2)A_j)
\end{align}
measures the Frobenius mass discarded by applying the cutoff to the low eigenvalues of $A_j$.

Crucially, a small dichotomy argument applies to these errors. We will prove that if $\rho\in \Sep(d,k)$ then all $e_j$ are bounded by $O(\eta^2/k)$, while if $\rho$ is $\eta$-far from every separable state in Frobenius norm then either there are some large $e_j$ (and thus we can reject since this cannot happen if $\rho\in \Sep(d,k)$) or $e_j$ are small and the dimensionality reduction machinery produces a safe state on which to use our algorithm of \Cref{sec:frombakshi}. This is why in \Cref{alg:dimreduction} (\Cref{line:firsttest}) we can safely reject every time we encounter a large error. 

This motivates implementing the filter $\Phi_j$ (\Cref{alg:dimreduction} (\Cref{line:filtering})), the quantum channel that implements a soft projection onto the important part of the eigenspace $F_j$. Formally, $\Phi_j:\cL(\C^d)\to \cL(\C^{d+1})$ is defined as 
\[
\Phi_j(X) = F_jX F_j+\Tr((I-F_j^2)X)\ket{\bot}\bra{\bot}.
\]

Then a simple counting argument tells us that the local dimensions, after applying the $\Phi_j$'s locally, is bounded. Indeed, every eigenvector retained by $F_j$ corresponds to an eigenvalue of $A_j$ with magnitude at least $\Omega(\eta^4/k^2)$. However, we already observed that the sum of the eigenvalues of $A_j$ is at most $1$ and thus there are at most $q=O(k^2/\eta^4)$ such eigenvectors and $\textnormal{rank}(F_j)\le q$. Equivalently, it also holds that
\[
\textnormal{rank}(\Tr_{-j}(\Phi_1\otimes\cdots\otimes\Phi_k)\rho)\le q,
\]
if we think about the quantum channel $\Phi_j$ instead. Thus we know that after applying the filters $\Phi_j$, there should be a low-dimensional representation $\widetilde \rho\in\cD((\C^q)^{\otimes k})$ of $(\Phi_1\otimes\cdots\otimes\Phi_k)(\rho)\in\cD((\C^{d+1})^{\otimes k})$. More formally, there are isometries $V_j:\C^q\to\C^{d+1}$ such that 
\[
(\Phi_1\otimes\cdots\otimes\Phi_k)(\rho) = \left(V_1\otimes \cdots\otimes V_k\right)\widetilde \rho (V_1^\dagger\otimes \cdots\otimes V_k^\dagger).
\]

However, the biggest conceptual obstacle now is that the argument leading to the isometries $V_j$ was completely existential, and learning them explicitly would require at least $\poly(d)$ time. Up to this point, the ideas more or less mimic the spirit of the procedure we devised in \Cref{part:classical}. The genuinely quantum ingredient is that we can now exploit the existence of the low-dimensional support without ever constructing a classical description of it.

Indeed, to do the last step, we resort to a general procedure which takes copies of an unknown state $\rho\in \cD((\C^d)^{\otimes k})$, which secretly lives in a space with small local dimension $q\ll d$, and extracts from it a state with this small local dimension. We rely on representation theory and more specifically on the Schur–Weyl duality (see \cite[Sections~2.1-2.3]{burchardt2025high} and \cite[Corollary 5.19.2-Theorem 5.22.1]{etingof2011introduction}). Indeed, writing down an explicit basis for the small space would require $\Omega(qd)$ space, but we want to do it in polynomial in $O(\log d)$ time. The insight is that we do not care about the rotation of the local states. Indeed, applying a local rotation to every factor would not change the distance to the set of separable states.
Let's focus on a local factor $j$, and consider the $N$ copies supplied. This gives us a state embedded in $\cD((\C^{d})^{\otimes N})$ as $V^{\otimes N}X (V^\dagger)^{\otimes N}$ where $X\in\cD((\C^{q})^{\otimes N})$ for some unknown isometry $V:\C^q\to \C^d$. For the sake of clarity, imagine the simplified case in which the Schur-Weyl duality tells us that
\[
(\C^m)^{\otimes N}\cong Q^{(m)}\otimes P,
\]
and such that when we apply the same unitary $U\in U(m)$ to every local factor it acts only on $Q^{(m)}$, while permutations act only on $P$ (in reality Schur-Weyl duality tells us $(\C^m)^{\otimes N}$ decomposes into the direct sum of terms like $Q_\lambda^{(m)}\otimes P_\lambda$, one for each $\lambda$ that represents a Young diagram with at most $m$ rows; see \Cref{sec:rep-theory} for additional details).
More importantly, the isometry $V^{\otimes N}$ takes the form $V_Q\otimes I_P$ for some induced isometry $V_Q:Q^{(q)}\to Q^{(d)}$.
So if $X\in \cD(Q^{(q)}\otimes P)$, then its embedding into $\cD(Q^{(d)}\otimes P)$ is
\[
(V_Q\otimes I_P) X (V_Q^\dagger\otimes I_P).
\]
So the information about the unknown isometry $V$ is confined to the register $Q^{(m)}$. So we can trace out the registers of $Q^{(m)}$ to obtain 
\[
\tau_P=\Tr_{Q^{(d)}}((V_Q\otimes I_P)X(V^\dagger_Q\otimes I_P))=\Tr_{Q^{(q)}}(X)
\]
and then, preparing an arbitrary state $\ket{t}\bra{t}_{Q^{(q)}}$, we can obtain $Z=\ket{t}\bra{t}_{Q^{(q)}}\otimes \tau_P$. We destroyed all the information contained in $Q^{(q)}$. This is not problematic as arbitrary rotations would act in the same way, while keeping the same distance from separable states. Formally, consider the Haar averages
\[
\cT(X)=\int_{U(q)} U^{\otimes N} X(U^\dagger)^{\otimes N} dU,
\]
and, using the fact that local rotations act only on the $Q$ registers, we obtain $\cT(X)=\frac{I_{Q^{(q)}}}{\dim Q^{(q)}}\otimes \tau_P$. On the other hand, the same happens with $\cT(Z)$ and thus even if $Z\neq X$ we obtain that $\cT(Z)=\cT(X)$. Indeed, we do not expect to recover $X$ exactly, since it would require $\poly(q^N)$ space to write it down. The procedure discussed here is implemented in \Cref{alg:dimreduction} (\Cref{line:compression}).

The discussion above concerned a single local factor. Importantly, the same identity remains valid when its $N$ registers are entangled with an arbitrary reference system. We may therefore apply this procedure independently to each of the $k$ local factors of the full state $\widetilde\rho^{\otimes N}$ and find a state $Z$ such that:\footnote{Note that we are choosing a single rotation $U=U_1\otimes \cdots\otimes  U_k$ for all copies, since the low-dimensional tester expects the same state over all copies.}
\[
(\cT_1\otimes \cdots \otimes \cT_k)(Z) = (\cT_1\otimes \cdots \otimes \cT_k)(\widetilde \rho^{\otimes N})=\mathbb{E}_{U_1,\ldots, U_k}[(U\widetilde\rho U^\dagger)^{\otimes N}].
\] 
This is sufficient because for every fixed $U$, $(U\widetilde\rho U^\dagger)^{\otimes N}$ and $\widetilde\rho^{\otimes N}$ satisfy the same separability promise; thus, the low-dimensional tester, by linearity, preserves the same acceptance probability gap under averaging.

\subsection{Implementing the Channels} \label{sec:channels}%

After the high-level overview, we now present the algorithm's implementation. As standard, we say that a channel $\widehat{\mathcal E}$ implements the ideal channel $\mathcal E$ with error at most $\delta$ if, for every reference register $B$ and every joint state $X$ on the input register and $B$,
\[
\left\|((\widehat{\mathcal E}-\mathcal E)\otimes I_B)X\right\|_1\le\delta.
\]
We assume the following subroutines (the proofs of these lemmas are deferred to \Cref{app:subroutines}).
The main difficulty is that $F_j=f(A_j)$ depends on the unknown operator $A_j=\Tr_{-j}(\rho^2)$. The following lemma shows that the corresponding spectral filter can nevertheless be implemented directly from copies of $\rho$, without explicitly learning $A_j$, using existing Hamiltonian-simulation techniques \citep{kimmel2017hamiltonian}.

\begin{restatable}[Filtering]{lemma}{lemmadrfilter}\label{lem:dr-filter}
For $j\in[k]$, define the channel $\Phi_j:\cL(\C^d)\to\cL(\C^d\oplus\textnormal{Span}\{\ket\bot\})$ by
\begin{equation}\label{eq:dr-filter-channel}
\Phi_j(X)=F_jXF_j+\Tr\left[\bigl(I-F_j^2\bigr)X\right]
\ket\bot\bra\bot.
\end{equation}
For $\delta>0$, a routine implements $\Phi_j$ with error at most $\delta$ on a specified target register, using $\poly(k,\eta^{-1},\delta^{-1})$ auxiliary copies of $\rho$ and  $\poly(k,\eta^{-1},\delta^{-1},\log d)$ total time.%
\end{restatable}

Once the filtering channel is available, estimating the discarded Frobenius mass $e_j$ is essentially standard: a swap-test-type measurement turns $e_j$ into the expectation of a bounded random variable, to which concentration applies.

\begin{restatable}[Estimation]{lemma}{lemmadresimate}\label{lem:dr-estimate}
Let $e_j$ be defined as in \Cref{eq:errordimred}. Given copies of $\rho$, $\epsilon>0$, and a probability $p>0$, there is a routine that returns classical estimates satisfying
\[
\max_{j\in[k]}|\widehat e_j-e_j|\le\epsilon,
\]
with probability at least $1-p$. It uses $\poly(k,\eta^{-1},\epsilon^{-1},\log 1/p)$ copies and this bound times $\poly(\log d)$ total time.
\end{restatable}

This is the key representation-theoretic ingredient: it removes the unknown embedding of a $q$-dimensional subspace inside $\mathbb C^d$ without ever learning that subspace. The price to pay is only that the output is preserved modulo a common local Haar average, which is sufficient for our application.

\begin{restatable}[Schur compression]{lemma}{lemmadrschurcompression}\label{lem:dr-schur}
For integers $D\ge q$ and $N$, there is a channel
\[
\mathcal C_{D\to q}^{(N)}:
\cD((\C^D)^{\otimes N})\longrightarrow\cD((\C^q)^{\otimes N})
\]
such that, for every reference system $B$, every isometry $V:\C^q\to\C^D$ and state $X\in\cD((\C^q)^{\otimes N}\otimes B)$,
\begin{equation}\label{eq:dr-schur-average}
\left(\big(\cT_q^{(N)}\circ\cC_{D\to q}^{(N)}\big) \otimes I_B\right)((V^{\otimes N}\otimes I_B)X((V^\dagger)^{\otimes N}\otimes I_B))=(\mathcal T_q^{(N)}\otimes I_B)(X),
\end{equation}
where, with $dU$ the Haar measure of $U(q)$, we let
\[
\cT_q^{(N)}(X):=\int_{U(q)}U^{\otimes N}X(U^\dagger )^{\otimes N}dU.
\]
The channel depends only on $D,q,N$ and uses no additional input copies. For $\delta>0$, it can be implemented with error at most $\delta$ in time $\poly(N,q,\log D,\log(1/\delta))$.
\end{restatable}

Finally, implementing the Haar average is standard: it suffices to efficiently sample an approximate Haar-random unitary and apply the same sampled unitary to all $N$ copies.

\begin{restatable}[Sampling Local Unitaries \citep{mezzadri2006generate}]{lemma}{lemmadrhaar}\label{lem:dr-haar}
For integers $q,N$, and $\delta>0$, a classical
randomized routine returns a unitary circuit $\widetilde U$ such that
the averaged channel
\[
X\mapsto\mathbb E_{\widetilde U}\left[\widetilde U^{\otimes N}X (\widetilde U^\dagger)^{\otimes N}\right]
\]
implements $\mathcal T_q^{(N)}$ with error at most $\delta$. Its random bits, classical time, and circuit size are $\poly(q,\log(N/\delta))$. On qubit registers, the comparison is on the full register, extending each ideal $U\in U(q)$ by the identity on unused basis states.
\end{restatable}

With these implementations of the ideal channels, we are ready to present the main algorithm.

\begin{algorithm}[!t]
\caption{Dimensionality Reduction}
\label{alg:dimreduction}
\begin{algorithmic}[1]
    \Require Copies of $\rho\in\cD((\C^d)^{\otimes k})$ and a low-dimensional tester $\cA$ for $\widetilde\rho\in\cD((\C^q)^{\otimes k})$
    \State For all $j\in [k]$ estimate $e_j$ with precision $\epsilon=O(\eta^2/k)$ and constant probability to obtain $\widehat e_j$ with \Cref{lem:dr-estimate}
    \If{$\max_{j\in[k]}\widehat e_j\ge\Omega(\eta^2/k)$}\label{line:firsttest}
        \State\textbf{reject}
    \EndIf
    \State Assign the $N$ copies of $\rho$ to registers $Y_1^{(n)},\ldots, Y_{k}^{(n)}$ for $n\in[N]$.
    \State For each local factor $j$ and any of the $N$ copies apply $\Phi_j$ as in \Cref{lem:dr-filter} (using additional fresh copies for its execution) with error $\delta=O(1/(kN))$\label{line:filtering} 
    \State Apply the compression channel $\cC^{(N)}_{d+1\to q}$ of \Cref{lem:dr-schur} to each local factor $j$ on the registers $Y_j^{(1)},\ldots,Y_j^{(N)}\in \C^{d+1}$ with error $O(1/k)$ and label the output registers as $Z_j^{(1)},\ldots,Z_j^{(N)}\in \C^q$ \label{line:compression}
    \State For each local factor $j$ apply \Cref{lem:dr-haar} independently with error $O(1/k)$, a Haar-random unitary $U_j\in U(q)$. Apply the same $U_j$ to all $N$ registers $Z_j^{(1)},\ldots,Z_j^{(N)}$ \label{line:haar}
    \State Regroup the registers $Z_1^{(n)},\ldots,Z_k^{(n)}$ by copy $n\in [N]$ and run the low-dimensional tester $\cA$
    \Return $\{\textbf{reject},\textbf{accept}\}$ according to the answer of $\cA$ \label{line:lowdimensionaltest}
\end{algorithmic}
\end{algorithm}
\subsection{Proof of \Cref{th:dimension-reduction}}

First we argue that separable inputs pass step 1 of the algorithm.
Suppose $\rho\in\Sep(d,k)$.
Fix any $j\in[k]$ and write with $x_a\in\C^d$ the $j$-th register and $y_a\in(\C^d)^{\otimes(k-1)}$ the remaining ones, so that
\[
\rho=\sum_a p_a\ket{x_a\otimes y_a}\bra{ x_a\otimes y_a},
\quad p_a\ge0,\quad \sum_a p_a=1,\quad
\|x_a\|=\|y_a\|=1.
\]
For the actual implementation of $f$ we pick $f=f_a$ with $a=2^{11}k^2/\eta^4$. %
Since $f(t)=1$ for $t\ge\eta^4/(2^{10}k^2)$, it is simple to check that for all $t\in(0,1)$
\begin{align}\label{eq:tmp4872}
t(1-f(t)^2)^2\le\frac{\eta^4}{2^{10}k^2}.
\end{align}
Consequently, plugging in the expansion of $\rho$ into the definition of $e_j:=\Tr((I-F_j^2)A_j) $ we obtain
\begin{align*}
    e_j&= \sum_a p_a\Re[\langle x_a\otimes y_a|((I_j-F_j^2)\otimes I_{-j})\rho\ket{x_a\otimes y_a}]\\
    &\le\sum_a p_a
\left\|\bra{x_a}((I_j-F_j^2)\otimes I_{-j})
\rho\right\|\\
&=\sum_a p_a\sqrt{\bra{x_a}(I_j-F_j^2)^2{A_j\ket{x_a}}}
\le\frac{\eta^2}{32k},
\end{align*}
where in the last inequality we used \Cref{eq:tmp4872} applied to the eigenvalues. Thus, whenever all estimates have the requested accuracy, we know that
\(
\widehat e_j\le O(\eta^2/k),
\)
and we know that observing a larger $\widehat e_j$ means that we can safely reject. If otherwise the algorithm does not reject at that step we know that $e_j\le O(\eta^2/k)$. Next, we argue that the filtering step preserves the distance to the set of separable states.

Define the projection onto the original registers by
\(
P:=(I_{\C^d}\oplus0_{\ket{\bot}})^{\otimes k}.
\)
Identifying the range of $P$ with $(\C^d)^{\otimes k}$ and defining
\[
\sigma = (\Phi_1\otimes \cdots\otimes \Phi_k)(\rho)
\]
the channel formula
of \Cref{eq:dr-filter-channel} in \Cref{lem:dr-filter} gives
\[
P\sigma P=\big(\bigotimes_j F_j\big)\rho\big(\bigotimes_j F_j\big).
\]
The inequalities $0\preceq F_j\preceq I$ give
\begin{align*}
\Big\|\big(I-\bigotimes_j F_j\big)\rho\Big\|^2
&=\Tr\left[\big(I-\bigotimes_j F_j\big)^2\rho^2\right]\\
&\le\Tr\left[\big(I-\bigotimes_j F_j^2\big)\rho^2\right]\tag{$0\preceq \otimes_j F_j\preceq I$}\\
&\le\sum_{j}\Tr\left[((I-F_j^2)\otimes I_{-j})\rho^2\right]\tag{$I-\otimes_{j}F^2_j\preceq\sum_{j}(I-F_j^2)\otimes I_{-j}$} \\
&=\sum_j\Tr\left[\bigl(I-F_j^2\bigr)A_j\right]=\sum_j e_j.
\end{align*}
Now, by triangle inequality, we can bound the error between $\rho$ and the application of the filtering $\Phi_j$:
\begin{equation}\label{eq:dr-filter-close}
\begin{aligned}
\|\rho-P\sigma P\|
&\le
\Big\|\big(I-\bigotimes_j F_j\big)\rho\Big\|+\Big\|\big(\bigotimes_j F_j\big)\rho
\big(I-\bigotimes_j F_j\big)\Big\|\\
&\le2\Big\|\big(I-\bigotimes_j F_j\big)\rho\Big\|
\le2\sqrt{\sum_j e_j}\le O(\eta).
\end{aligned}
\end{equation}
Now we finally prove that $\sigma=(\Phi_1\otimes\cdots\otimes\Phi_k)(\rho)$ satisfies the same separability promises as $\rho$.
For all $\zeta\in\Sep(d+1,k)$, the operator $P\zeta P$, restricted to the original registers, is subnormalized and separable. Thus $0\le\Tr(P\zeta P)\le1$. Define the restricted separable state
\[
\tilde\zeta =P\zeta P+(1-\Tr(P\zeta P))\frac I{d^k}\in\Sep(d,k),
\]
and observe that
\[
\|\rho-\tilde\zeta\|^2=\|\rho-P\zeta P\|^2-(1-\Tr(P\zeta P))^2\frac{1}{d^k},
\]
which implies that
\begin{align*}
\dist(\rho,\Sep(d,k))&\le \|\rho-\tilde\zeta\|\\
&\le\|\rho-P\zeta P\|\\
&\le \|\rho-P \sigma P\|+\|P(\sigma- \zeta) P\|\\
&\le O(\eta)+\|\sigma-\zeta\|.
\end{align*}
Minimizing over $\zeta\in\Sep(d+1,k)$ proves
\begin{equation}
\dist(\sigma,\Sep(d+1,k))
\ge\dist(\rho,\Sep(d,k))-O(\eta).
\end{equation}
On the other hand, since each $\Phi_j$ is a local channel, $(\Phi_1\otimes\cdots\otimes\Phi_k)(\Sep(d,k))
\subseteq\Sep(d+1,k)$ and therefore
 \[
\rho\in\Sep(d,k)\implies
\sigma\in\Sep(d+1,k).
\]
Finally, we connect the supplied separability tester with the last two subroutines. More precisely we show that there are unknown local isometries such that there is a small state $\tilde\rho$ which satisfies the same separability promises of $\rho$. We do this by first proving that the rank of the filtering step is small.
Since $f(t)>0$ precisely when $t>\eta^4/(2^{11}k^2)$,
\[
\frac{\eta^4}{2^{11}k^2}\operatorname{rank}F_j
<\Tr(A_j)=\Tr(\rho^2)\le1.
\]
The definition of $\Phi_j$ in \Cref{eq:dr-filter-channel} also gives
\[
\begin{aligned}
\operatorname{supp}\Tr_{-j}
(\sigma)\subseteq\operatorname{img}F_j
\oplus\text{Span}\{\ket\bot\},
\end{aligned}
\]
and therefore
\[
\begin{aligned}
\operatorname{rank}\Tr_{-j}(\sigma)\le1+\operatorname{rank}F_j\le O(k^2/\eta^4)=q.
\end{aligned}
\]
Thus we know that there are isometries $V_j:\C^q\to\C^{d+1}$ with $V_j^{\dagger}V_j=I$ such that for all $j\in[k]$:
\[
\operatorname{supp}\Tr_{-j}(\sigma)\subseteq\operatorname{img}V_j
\]

and define,
\[
\widetilde\rho:=
\Big(\bigotimes_jV_j^{\dagger}\Big)\sigma\Big(\bigotimes_jV_j\Big).
\]

This choice of $V_j$ implies
\(
\sigma=\Big(\bigotimes_jV_j\Big)\widetilde\rho\Big(\bigotimes_jV_j^{\dagger}\Big),
\)
and
\(
\Tr\widetilde\rho=%
1.
\)
Therefore
\[
\sigma\in\Sep(d+1,k)\implies\widetilde\rho\in\Sep(q,k).
\]

Now we show that the distance of $\tilde\rho$ from separable states is preserved. Indeed, for every $\sigma'\in\Sep(q,k)$,
\[
\Big(\bigotimes_jV_j\Big)\sigma'\Big(\bigotimes_jV_j^{\dagger}\Big)
\in\Sep(d+1,k),
\]
and
\[
\|\widetilde\rho-\sigma'\|
=\Big\|\sigma-
\Big(\bigotimes_jV_j\Big)\sigma'\Big(\bigotimes_jV_j^{\dagger}\Big)
\Big\|.
\]

Which implies that,
\[
\dist(\widetilde\rho,\Sep(q,k))
\ge\dist(\sigma,\Sep(d+1,k)).
\]

Combined with the relationship between the separability of $\rho$ and $\sigma$ we obtain that
\begin{align}
&\rho\in\Sep(d,k)\implies\sigma\in\Sep(d+1,k)\implies\widetilde\rho\in\Sep(q,k)\notag\\
&\hspace{5cm}\text{and}\label{eq:dr-distance}\\
&\dist(\widetilde \rho,\Sep(q,k))\ge\dist(\sigma,\Sep(d+1,k))-O(\eta)\ge\dist(\rho,\Sep(d,k))-O(\eta)\notag.
\end{align}

The last step concerns the application of the Haar averaging together with the compression step of \Cref{lem:dr-schur} after regrouping the register for each party. Applying the compression channel $\cC_{d+1\to q}^{(N)}$ and the ideal Haar averaging channel to each party independently, we obtain
\begin{align*}
    \left(\bigotimes_{j\in[k]}\Big(\cT_q^{(N)}\circ \cC_{d+1\to q}^{(N)}\Big)\right)(\sigma^{\otimes N}) &= \left(\bigotimes_{j\in[k]}\cT_q^{(N)}\right)(\widetilde\rho^{\otimes N})\\
    &=\int_{U(q)^k}\left[(\otimes_{j\in[k]} U_j)\widetilde \rho  (\otimes_{j\in[k]} U_j^\dagger)\right]^{\otimes N} dU_1\cdots dU_k.
\end{align*}
Moreover, since local unitaries preserve separability:
\[
\dist\left((\otimes_{j\in[k]} U_j)\widetilde \rho  (\otimes_{j\in[k]} U_j^\dagger),\Sep(q,k)\right) = \dist\left(\widetilde \rho  ,\Sep(q,k)\right),
\]
and thus every term in the integrand preserves distance from separable states and by \Cref{eq:dr-distance} the low-dimensional tester $\cA$ on $\left[(\otimes_{j\in[k]} U_j)\widetilde \rho  (\otimes_{j\in[k]} U_j^\dagger)\right]^{\otimes N}$, returns the correct answer with constant probability. %
The implementations of the ideal channels $\cC_{d+1\to q}^{(N)}$ and $\cT_q^{(N)}$ are discussed in \Cref{lem:dr-schur} and \Cref{lem:dr-haar}, respectively. 

\paragraph{Resource Management.}

Finally, we bound the error and the total resources used by the algorithm. By \Cref{lem:dr-estimate}, 
\[
\max_j|\widehat e_j-e_j|<O(\eta^2/k),
\]
holds with constant probability.
All later target and auxiliary copies are fresh and independent of
these estimates. Let $\widehat\Omega$ be the joint output of the $kN$
filtering calls. Applying \Cref{lem:dr-filter} successively, with all
other registers as reference registers, gives
\[
\|\widehat\Omega-\sigma^{\otimes N}\|_1\le O(1).
\]
The two remaining contributions follow from \Cref{lem:dr-schur,lem:dr-haar}. The total change in acceptance probability, including estimation failure, is at most constant.
Consequently, choosing the constants accordingly, if $\rho\in\Sep(d,k)$ the probability of accepting can be made at least $2/3$, while making the probability of accepting in the case $\dist(\rho,\Sep(d,k))\ge\eta$ at most $1/3$.

Substituting $q=O(k^2/\eta^4)$ and the requested accuracies in the four
subroutine bounds gives $\poly(k,N,\eta^{-1})$ copies and total time
$T+\poly(k,N,\eta^{-1},\log d)$; $\mathcal A$ is only run once.

\printbibliography
\newpage
\appendix
\section{Proofs Omitted from \Cref{sec:frombakshi}}

\subspacetomography*
\begin{proof}
    Let $F:\C^D\to (\C^q)^{\otimes m}$ be defined as $F\ket\ell=\ket{a_\ell}$. Then $F^\dagger F = I_D$ and $FF^\dagger =P$, since the $a_\ell$ are orthonormal. Moreover, because vectors $a_\ell$ are explicitly given, the isometry $F$, together with the projection onto its image, can be implemented in time $\poly(m,q,D)$.
    
    Apply standard Frobenius tomography (such as the one of \citet{flammia2024quantum}) to $\omega'=(F^\dagger \omega F)\oplus(1-\Tr(P\omega))\ket{\bot}\bra{\bot}$, which is a density matrix of $\C^D\oplus\textnormal{Span}\{\ket\bot\}$ since $\Tr(F^\dagger \omega F)=\Tr(FF^\dagger\omega)=\Tr(P\omega)$. So we get back a classical description $\widehat \omega'$ such that 
    \[
    \|\omega'-\widehat\omega'\|\le\epsilon,
    \]
    in time and copies $O(D^2\epsilon^{-2}\log 1/\delta)$. Let $Q$ be the projector from $\C^D\oplus\textnormal{Span}(\ket\bot)$ to $\C^D$, and define $\widehat\omega=FQ\widehat\omega'QF^\dagger$. Since $FQ\omega'QF^\dagger = F(F^\dagger\omega F)F^\dagger=P\omega P$,
    \begin{align*}
        \|\widehat \omega-P\omega P\|& = \|FQ(\widehat \omega'-\omega' )QF^\dagger\|\\
        &\le \|Q(\widehat \omega'-\omega') Q\|\\
        &\le\|\omega'-\widehat\omega'\|\le \epsilon,
    \end{align*} %
    where we used that the Frobenius norm is nonexpansive under isometries and projections, thus
    concluding the proof.
\end{proof}

\lemmapolynomialopt*

\begin{proof}
    Fix $\tau=O(\epsilon/D^2)$ and $\gamma=O(\epsilon/D)$. For every $j\in[D]$, we flatten the tensor $T_j\in E^{\otimes j}$ along the $s\in[j]$ coordinate giving the linear map 
    \[
    A_{j,s}:E\to E^{\otimes(j-1)},
    \]
    and let $A=\bigoplus_{j\in[D]}\bigoplus_{s\in[j]} A_{j,s}$, which is a linear map $A:E\to \bigoplus_{j\in[D]}\bigoplus_{s\in[j]}E^{\otimes(j-1)}$.

    The Frobenius norm of $A$ can be upper bounded as
    \[
    \|A\|^2=\sum_{j=1}^D\sum_{s\in[j]}\|A_{j,s}\|^2\le \sum_{j=1}^Dj\|T_j\|^2\le D.
    \]

    Now we can apply \Cref{lem:onecutcompression} with accuracy $\tau$ that returns a subspace $W\subseteq E$ such that $\dim(W)\le (D/\tau)^2$ and $\|A(I-P_W)\|_{op}\le \tau$. Now let $Q_{j,s}$ be the orthogonal projection of $A$ onto the $(j,s)$ direct sum component of $A$, i.e., $A_{j,s}= Q_{j,s}A$. Then
    \begin{align*}
        \|A_{j,s}(I-P_W)\|_{op}&=\|Q_{j,s}A(I-P_W)\|_{op}\\
        &\le \|A(I-P_W)\|_{op}\\
        &\le \tau,
    \end{align*}
    since $Q_{j,s}$ is a projection and thus $\|Q_{j,s}\|_{op}\le1$.
    Consequently $y\mapsto p(y)$ is uniformly approximated by the uniform projection of $y$ to $W$ as %
    \begin{align}
        |p(y)-p(P_Wy)|&\le\sum_{j=0}^D|\ip{T_j}{(I-P^{\otimes j}_W)y^{\otimes j}}|\notag\\
        &\le \sum_{j=0}^D\sum_{s=1}^j \|A_{j,s}(I-P_W)\|_{op}\le \tau D^2\label{eq:tmp846}.
    \end{align}

    Let $y^*$ be the maximizer of $e^{-\phi(y)}p(y)$ and let $\cG$ be a $\gamma$-net of $B_W$ (in $\ell_2$-norm). Since $\dim(W)\le (D/\tau)^2$, we know that $|\cG|\le \gamma^{-O((D/\tau)^2)}$. Thus there exists a $g^*\in \cG$ such that $\|g^*-P_Wy^*\|\le \gamma$.

    For every $g\in \cG$ solve the following optimization problem
    \begin{align*}
        &\min\phi(y)\\
        &\text{s.t.}\quad y\in \cK,\|P_Wy-g\|\le \gamma,
    \end{align*}
    and obtain $y_g$ if the problem is feasible. Observe that the problem is feasible for $g=g^*$ (with feasible point $y^*$) and $\phi(y_{g^*})\le \phi(y^*)$. By the triangle inequality, we also have that
    \begin{align}\label{eq:tmptriang}
    \|P_W y_{g^*}-P_Wy^*\|\le \|P_W y_{g^*}-g^*\|+\|g^*-P_Wy^*\|\le 2\gamma.
    \end{align}

    Next we show that $p$ is $D$-Lipschitz. Indeed, we can write 
    \[
    y^{\otimes j}-(y')^{\otimes j}=\sum_{s=1}^j y^{\otimes (s-1)}\otimes(y-y')\otimes (y')^{\otimes (j-s)},
    \]
    and thus 
    \begin{align*}
       |p(y)-p(y')|&\le\sum_{j=0}^D|\ip{T_j}{y^{\otimes j}-(y')^{\otimes j}}|\\
       &\le \sum_{j=0}^D \sum_{s=1}^j \|T_j\|\|y-y'\|\\
       &\le \sum_{j=0}^D j \|T_j\|\|y-y'\|\\
       &\le D \|y-y'\|.
    \end{align*}
    This lets us conclude easily that $p(y_{g^*})$ is close to $p(y^*)$. 
    \begin{align}
        |p(y_{g^*})-p(y^*)|&\le |p(y_{g^*})-p(P_Wy_{g^*})|+|p(P_Wy_{g^*})-p(P_Wy^*)|+|p(P_Wy^*)-p(y^*)|\notag\\
        &\le \tau D^2+D\|P_W(y_{g^*}-y^*)\|+\tau D^2\tag{\Cref{eq:tmp846}}\\
        &\le 2\tau D^2+2D\gamma\tag{\Cref{eq:tmptriang}}\\
        &\le O(\epsilon)\label{eq:tmp478}
    \end{align}

    Moreover, since $\OPT\ge \epsilon$ and $\phi(y)\ge0$ we also obtain that $p(y^*)\ge \epsilon$, and \Cref{eq:tmp478} implies that $p(y_{g^*})\ge p(y^*)- O(\epsilon)>0$. We can combine this with the already observed fact that: $e^{-\phi(y_{g^*})}\ge e^{-\phi(y^*)}$ to conclude that
    \begin{align*}
        e^{-\phi(y_{g^*})}p(y_{g^*}) \ge e^{-\phi(y_{g^*})}p(y^*)-O(\epsilon) \ge \OPT-O(\epsilon),
    \end{align*}
    showing that there is a $g\in \cG$ which returns a $y_g$ with objective greater than $\OPT-O(\epsilon)$. The proof is concluded by showing that $|\cG|\le (\epsilon/D)^{-O((D^3/\epsilon)^2)}=O_{\epsilon,D}(1)$ and that every execution of the convex program is polynomial in the dimension of $E$.
\end{proof}

\lemmaLS*

\begin{proof}

   We are going to divide the proof into several steps.

   Then define $G=1+\sum_{\ell=1}^T|c_\ell|$ and note that $\|\omega-K\|_{op}\le 1+\sum_{\ell=1}^T|c_\ell|:=G$.
    
   \paragraph{Relevant $x$ are Represented in $\cK_S$.}
    
   Let $x=\otimes_{i\in[m]}x_i\in R(a,\alpha)$ and write it as
   \[
   x_i=\frac{a_i+z_i}{\sqrt{1+\|z_i\|^2}},
   \]
   for all $i<m$, where $z_i\in a_i^\bot$. The fact that $x\in R(a,\alpha)$ means that it cannot be that some $x_i$ is orthogonal to $a_i$, otherwise we would have that $\ip{a}{x_{[m-1]}}=0$ contradicting that $x\in R(a,\alpha)$.

   Moreover, since $x\in R(a,\alpha)$, we get the following implications
   \begin{align*}
       |\ip{a}{x_{[m-1]}}|^2&=\prod_{i<m} |\ip{a_i}{x_i}|^2\\
       &= \prod_{i<m}\frac{1}{1+\|z_i\|^2}\ge \alpha,
   \end{align*}
   and thus $\sum_{i<m}\log(1+\|z_i\|^2)\le \log(1/\alpha)$. Define $S=\{i<m:\|z_i\|>1\}$. Every index $i\in S$ contributes at least $\log(2)$ and thus there can be at most $\log(1/\alpha)/\log(2)\le2\log(1/\alpha)$.

   On the other hand for all $i\in[m-1]\setminus S$ we have $\|z_i\|\le 1$ and thus
   \[
   \sum_{i\in [m-1]\setminus S} \|z_i\|^2\le 2\sum_{i\in [m-1]\setminus S} \log(1+\|z_i\|^2)\le 2\log(1/\alpha),
   \]
   and thus, setting $u_i=x_i$ for all $i\in S\cup\{m\}$ we have that $(z,u)\in \cK_S$ for some $S\subseteq[m-1]$ and $|S|\le 2\log(1/\alpha)$.

   \paragraph{Distance from the Subspace $\cV_S$}

   Fix any $S$ and let $(z,u)\in \cK_S$ with $\|u_i\|=1$, and consider the normalized product vector
   \[
   x(z,u) = \bigotimes_{i\in[m-1]\setminus S}\frac{a_i+z_i}{\sqrt{1+\|z_i\|^2}}\otimes \bigotimes_{i\in S\cup\{m\}} u_i.
   \]
   If we expand the first tensor product, we can write it as:
   \[
    x(z,u) = \frac{1}{\prod_{i\in[m-1]\setminus S}\sqrt{1+\|z_i\|^2}}\sum_{R\subseteq[m-1]\setminus S}\left(\bigotimes_{i\in ([m-1]\setminus S)\setminus R}a_i\otimes \bigotimes_{i\in R }z_i\otimes \bigotimes_{i\in S\cup\{m\}} u_i\right).
   \]
   The projection $P_S$ on the subspace $\cV_S^h$ precisely keeps all the components with $|R|\le h$, and thus
   \[
   \|(I-P_S)x(z,u)\|^2=\frac{1}{\prod_{i\in[m-1]\setminus S}({1+\|z_i\|^2})}\sum_{t>h}\sum_{R\subseteq[m-1]\setminus S:|R|=t}\prod_{i\in R}\|z_i\|^2\le \sum_{t>h}\sum_{R\subseteq[m-1]\setminus S:|R|=t}\prod_{i\in R}\|z_i\|^2.
   \]
   Now we can use the usual estimate $\sum_{R\subseteq[n]:|R|=t}\prod_{i\in R}s_i\le\frac{(\sum_{i\in[n]}s_i)^t}{t!}$, applied to $s_i=\|z_i\|^2$ on the set $[m-1]\setminus S$, and, recalling that $\sum_{i\in [m-1]\setminus S}\|z_i\|^2\le 2\log(1/\alpha)$
   \begin{align*}
   \|(I-P_S)x(z,u)\|^2&\le \sum_{t>h} \frac{(2\log(1/\alpha))^t}{t!}\\
   &\le \sum_{t>h} \left(\frac{2e\log(1/\alpha)}{t}\right)^t\\
   &\le \sum_{t>h} \left(\frac{1}{2}\right)^t\tag{$h\ge 4e\log(1/\alpha)$}\\
   &\le 2^{-h}\\
   &\le \left(\frac{\epsilon}{32 G}\right)^2\tag{$h>4\log(32 G/\epsilon)$}
   \end{align*}

   \paragraph{Classical Optimization Approximates Quantum Objective.}

   We condition on the event that subspace tomography returns a $\widehat\omega$ such that
   \begin{align}\label{eq:tomogtmp}
   \|\widehat \omega-P_S\omega P_S\|_{op}\le \|\widehat \omega-P_S\omega P_S\|\le  \epsilon/16,
   \end{align}
   which holds with high-probability.
We want to show that for all $(z,u)\in \cK_S$ with $\|u_i\|=1$ the value $\Tr((\omega-K)\ket{x}\bra{x})$, with $x=x(z,u)$, approximates $F_S(z,u)$:
    \begin{align*}
        |\Tr((\omega-K)(\ket{P_Sx}\bra{P_Sx}-\ket{x}\bra{x}))|&\le |\Tr((\omega-K)\ket{(P_S-I)x}\bra{P_Sx})|+|\Tr((\omega-K)\ket{x}\bra{(P_S-I)x})|\\
        &\le 2\|\omega-K\|_{op}\|(P_S-I)x\|\|x\|\\
        &\le 2G\|(P_S-I)x\|\\
        &\le \frac{\epsilon}{16},
    \end{align*}
    where in the first inequality we used that $\ket{a}\bra{a}-\ket{b}\bra{b}=\ket{a-b}\bra{a}+\ket{b}\bra{a-b}$

    Consider now $F_S(z,u)$. Recall that $x(z,u)=e^{-\Phi(z)/2}[\bigotimes_{i\in[m-1]\setminus S}(a_i+z_i)\otimes\bigotimes_{i\in S\cup \{m\}}u_i]$ since $e^{-\Phi(z)/2}=\frac{1}{\prod_{i\in[m-1]\setminus S}\sqrt{1+\|z_i\|^2}}$ and thus $P_Sx(z,u) = e^{-\Phi(z)/2}P_S w_S(z,u)$, whereas (recalling that $H_S=\widehat \omega-P_SKP_S$):
    \begin{align*}
        F_S(z,u) &= e^{-\Phi(z)} \Tr(H_S\ket{w_S(z,u)}\bra{w_S(z,u)})\\
        &= \Tr(H_S\ket{e^{-\Phi(z)/2}w_S(z,u)}\bra{e^{-\Phi(z)/2}w_S(z,u)})\\
        &= \Tr(H_S\ket{P_Sx(z,u)}\bra{P_Sx(z,u)}).
    \end{align*}
    Finally,
    \begin{align}
        |F_S(z,u) -\Tr((\omega-K)\ket{x}\bra{x})|&\le |F_S(z,u) -\Tr((\omega-K)\ket{P_Sx}\bra{P_Sx})|+\frac{\epsilon}{16}\notag{}\\
        &=|\Tr((\widehat\omega-P_SKP_S-(\omega-K))\ket{P_Sx}\bra{P_Sx})|+\frac{\epsilon}{16}\notag{}\\
        &=|\Tr((\widehat\omega-P_SKP_S-P_S(\omega-K)P_S)\ket{P_Sx}\bra{P_Sx})|+\frac{\epsilon}{16}\notag{}\\
        &\le \|\widehat \omega-P_S\omega P_S\|_{op}\|P_Sx\|^2+\frac{\epsilon}{16}\notag{}\\
        &\le \frac{\epsilon}{8}\label{eq:tmp55589}
    \end{align}
    where in the last inequality we used \Cref{eq:tomogtmp}.

    \paragraph{Correctness.} Let $x^*\in R(a,\alpha)$ be a maximizer of $\Tr((\omega-K)\ket x \bra x)$. For the representation above, we can say that there exists a set $S^*$ and $(z^*,u^*)\in  \cK_{S^*}$ such that $x^*=x(z^*,u^*)$.

    Observe that we can apply \Cref{lem:classicopt} to $F_S$. We first rescale the variables $(z,u)\in \cK_S$ by $R=\sqrt{2\log(1/\alpha)+|S|+1}$, so that they are contained in the unit ball. The polynomial part of $F_S$ has degree at most $2(h+|S|+1)$. Moreover, using $\|H_S\|_F\le O(G)$, the sum of the Frobenius norms of its coefficient tensors, after rescaling, is bounded by $O_{\alpha,\epsilon,G}(1)$.
    Dividing the polynomial and the desired accuracy by this bound gives the required normalization.
    
    Let $(\hat z,\hat u)\in \cK_{S^*}$ be the values returned by the polynomial optimization subroutine of \Cref{lem:classicopt} and thus, using the previous steps of \Cref{eq:tmp55589}, we obtain
    \begin{align*}
        F_{S^*}(\hat z,\hat u)&\ge F_{S^*}( z^*,u^*)-\frac{\epsilon}{8}\\
        &\ge \Tr((\omega-K)\ket{x^*}\bra{x^*})-\frac{\epsilon}{4}.
    \end{align*}

    Observe that by assumption $\Tr((\omega-K)\ket{x^*}\bra{x^*})\ge \epsilon$ and thus $F_{S^*}(\hat z,\hat u)\ge\epsilon/2>0$. This means that the unnormalized value is itself strictly greater than zero, since $e^{-\Phi(z)}>0$, i.e., $\Tr(H_{S^*}\ket{w_{S^*}(\hat z,\hat u)}\bra{w_{S^*}(\hat z,\hat u)})>0$, but this implies that also all $\hat u_i$ are nonzero. We then build the physical vector $\hat x=x(\hat z, \tilde u)$, where $\tilde u=\otimes_{i\in S\cup\{m\}}\hat u_i/\|\hat u_i\|$ and, using again that the $F_S$ approximates the true value, we get
    \begin{align*}
        \Tr((\omega-K)\ket{\hat x}\bra{\hat x})&\ge F_{S^*}(\hat z, \tilde u)-\frac{\epsilon}{8} \\
        &\ge F_{S^*}(\hat z, \hat u)-\frac{\epsilon}{8}\tag{Homogeneity and $\|\hat u_i\|\le 1$}\\
        &\ge \Tr((\omega-K)\ket{x^*}\bra{x^*})-\frac{3\epsilon}{8}.
    \end{align*} 

    For each $S$, we compute $F_S(\hat z,\hat u)$, and if it exceeds $\epsilon/2$, we estimate it using copies of $\omega$ with accuracy $\epsilon/4$. Then we return the one with the highest value (including the fallback option $a\otimes b$ when all $F_S(\hat z,\hat u)$ do not pass the $\epsilon/2$ threshold).

    Assuming that we are in the good event, the optimum is greater than $\epsilon$, and all estimates hold, i.e., $|\Tr((\omega-K)\ket x\bra x)-v_x|\le \epsilon/4$ for all $x\in \cL$, then applying twice the fact that all estimates hold, we find that
    \[
    \Tr((\omega-K)\ket{\hat x}\bra{\hat x})\ge \Tr((\omega-K)\ket{ x^*}\bra{ x^*})-\epsilon.
    \]

    \paragraph{Probability and Resource Bounds.} The number of possible sets $S\subseteq[m-1]$ with $|S|\le 2\log(1/\alpha)$ is at most $M=m^{O(\log(1/\alpha))}$, and thus a union bound over all tomography events and all estimates gives a total success probability of at least $1-O(M\delta)$.

    On the other hand, the total number of copies required for tomography is $O(\max_{S}\dim(\cV_S^h)^2M\log(M/\delta)/\epsilon^2)$, and $O(M\log(M/\delta)/\epsilon^2)$ for the estimation step. It is not difficult to upper bound $\max_{S}\dim(\cV_S^h)$ with $O((mq)^{O(\log(2G/(\alpha\epsilon)))})$. Thus the total number of copies is upper bounded by
    \[
    O_{\alpha,\epsilon,G}(\poly(m q)\log(1/\delta)).
    \]
   
\end{proof}

\coverquantum*
\begin{proof}
    In the following, we assume that every call to $\LS$ gives the guarantees it can supply only with high probability. Each success has probability $1-\delta$, and in the $j$-th loop of the algorithm we make at most $O(\theta^{-2})$ calls. So the following holds on a good event, which holds with probability at least $1-O(k\delta\theta^{-2})$.
    \paragraph{Every Added $c\in C_j$ has High Overlap.} Suppose $c$ was added to $C_j$. This was done since there was an $a\in C_{j-1}$ such that $v_a\ge 3\theta/4$ and since, by definition of $\LS$ this was accurate up to $\theta/16$ we get that
    \begin{align}\label{eq:tmp098}
    \Tr((\rho^{[j]}-\varphi_j)\ket{c}\bra{c})\ge v_a-\frac{\theta}{16}\ge \frac\theta2,
    \end{align}
    which also implies that 
    \begin{align}\label{eq:tmp097}
        \Tr(\rho^{[j]}\ket{c}\bra{c})\ge \frac\theta2,
    \end{align}
    since $\Tr(\varphi_j\ket c\bra c)\ge 0$.

    \paragraph{Effect of the Regularizer $\varphi_j$.} Now we show that $\varphi_j$ pushes the element added to $C_j$ to be nearly orthogonal. Using \Cref{eq:tmp098} and the fact that $\rho^{[j]}$ is a density matrix, we can conclude that
    \begin{align*}
        1&\ge \Tr(\rho^{[j]}\ket{c}\bra{c})\tag{$\rho^{[j]}\in\cD((\C^d)^{\otimes j})$}\\
        &\ge \Tr(\varphi_j\ket{c}\bra{c}) \tag{\Cref{eq:tmp098}}\\
        &=\lambda\sum_{c'\in C_j}\Tr(\ket{c'}\bra{c'}\ket{c}\bra{c})\\
        &=\lambda\sum_{c'\in C_j}|\ip{c'}{c}|^2,
    \end{align*}
    from which it follows that each $c$, when added to the set $C_j$ we are building, satisfies that
    \begin{align}\label{eq:095}
    |\ip{c'}{c}|^2\le (\theta/8)^2.
    \end{align}

    \paragraph{Stopping Condition.}
    We now show that the cap of $t\le 8/\theta$ added in the inner loop of \Cref{alg:cover} is not reached on a good event and the $j$-th loop terminates because $\textsf{Stop}=\textsf{true}$. We prove this by contradiction. Suppose that for the first $N$ iterations of the inner loop we always had $\max_{a\in C_{j-1}}v_a\ge \frac34\theta$ and let $C_{j} = \{c_1,\ldots, c_N\}$. From \Cref{eq:tmp097} we get that
    \begin{align*}
        \frac{N\theta}{2}&<\sum_{t=1}^N\Tr(\rho^{[j]}\ket{c_t}\bra{c_t})\le  \left\|\sum_{t=1}^N\ket{c_t}\bra{c_t}\right\|_{op},
    \end{align*}
    We upper bound this by the standard fact that unit vectors with small overlap (\Cref{eq:095}) have small Gram operator norm.
    
    Define the Gram matrix $G$ such that $G_{st}=\ip{c_s}{c_t}$. Since $c_t$ are unit vectors, $G_{tt}=1$ and \Cref{eq:095} shows that $|\ip{c_t}{c_s}|\le \theta/8$ for $s\neq t$. Therefore, by the Gershgorin disk theorem (and since $G$ is Hermitian and its eigenvalues are real), all eigenvalues lie within a radius of $\sum_{s:s\neq t}|G_{ts}|\le \theta(N-1)/8$ of $G_{tt}=1$ and thus the maximum eigenvalue of $G$ is bounded by 
    \[
    1+\theta(N-1)/8.
    \]
    This is useful since $\sum_{t=1}^N\ket{c_t}\bra{c_t}$ and $G$ have the same spectrum (in the support) and $\|\cdot\|_{op}$ is just the norm of the largest eigenvalue since $\sum_{t=1}^N\ket{c_t}\bra{c_t}$ is Hermitian. This shows that 
    \[
    \frac{N\theta}{2}\le 1+\theta(N-1)/8,
    \]
    and rearranging we obtain $N\le 8/(3\theta)<8/\theta$.

    \paragraph{Coverage.} 
    Next we show that if $x=\otimes_{i\in[j]}x_i$ is a $j$-partite pure product state (with $\Tr(\rho^{[j]}\ket{x}\bra{x})\ge\theta$) that is not covered by $C_j$ then the conditions for $\LS$ would have been satisfied and there would have been an $\hat x_a$ such that $v_a\ge \frac34\theta$ contradicting the fact that (on a good event) the cap condition $t\le 8/\theta$ is not triggered. We do this by induction on $j$. The base case $j=0$ is trivial.

    First, we show that if $x$ has non-negligible overlap with $\rho^{[j]}$ then $x_{[j-1]}:=\otimes_{i\in[j-1]} x_i$ has non-negligible overlap with $\rho^{[j-1]}$. Indeed,
    \begin{align*}
        \Tr(\rho^{[j-1]}\ket{x_{[j-1]}}\bra{x_{[j-1]}})&= \Tr(\rho^{[j]}(\ket{x_{[j-1]}}\bra{x_{[j-1]}}\otimes I))\tag{By def.~of partial trace }\\
        &\ge \Tr(\rho^{[j]}(\ket{x_{[j-1]}}\bra{x_{[j-1]}}\otimes \ket{x_j}\bra{x_j})) \tag{$I\succeq \ket{x_j}\bra{x_j}$}\\
        &=\Tr(\rho^{[j]}\ket{x_{[j]}}\bra{x_{[j]}})\ge \theta.
    \end{align*}
    Then, by induction hypothesis, there exists $a\in C_{j-1}$ such that $|\ip{a}{x_{[j-1]}}|^2\ge \alpha$, which means that $x\in R(a,\alpha)$. 

    Next we show that if $x$ is not covered by $C_j$, then the value is high. Indeed, if $x$ is not covered by $C_j$, it means that  $|\ip{x}{c}|^2\le\alpha$ for all $c\in C_j$ and thus
    \[
    \Tr(\varphi_j\ket{x}\bra{x})=\lambda\sum_{c\in C_j}|\ip{c}{x}|^2\le\lambda|C_j|\alpha\le \left(\frac{8}{\theta}\right)^2 \frac{8}{\theta}\left(\frac{\theta}{8}\right)^4 = \frac\theta8,
    \]
    and thus
    \[
    \Tr((\rho^{[j]}-\varphi_j)\ket{x}\bra{x})\ge \frac{7}{8}\theta.
    \]
    Together with $x\in R(a,\alpha)$ this implies that
    \[
    v_a\ge \max_{x\in R(a,\alpha)}\Tr((\rho^{[j]}-\varphi_j)\ket{x}\bra{x})-{\frac{\theta}{8}}\ge {\frac34}\theta,
    \]
    contradicting the stopping condition that we proved in the previous step, which implies that $v_a<\frac{3}{4}\theta$ for all $a\in C_{j-1}$.

    Therefore, there must be a $c\in C_j$ such that $|\ip{x}{c}|^2\ge \alpha$. In particular, this holds for $j=k$, so $C_k$ is a $\theta$-overlap cover for $\rho$. On the other hand, if $C_{j-1}=\emptyset$, then by the induction hypothesis there are no high overlap $(j-1)$-partite product vectors and thus $\emptyset$ is itself a valid $\theta$-cover.
\end{proof}

\section{Representation Theory}\label{sec:rep-theory}

A unitary representation of a group $G$ on a finite-dimensional complex
space $E$ is a homomorphism $r:G\to U(E)$. It is irreducible if $E\ne\{0\}$
and its only invariant subspaces are $\{0\}$ and $E$. Representations $r,s$ are equivalent
if $Ar(g)=s(g)A$ for all $g$ and some invertible linear map $A$. We recall Schur's lemma (see \citet[Proposition~2.3.9]{etingof2011introduction}).

\begin{lemma}[Schur's lemma]\label{lem:rt-schur}
Let $r:G\to U(E)$ and $s:G\to U(F)$ be irreducible representations.
If $A:E\to F$ satisfies $Ar(g)=s(g)A$ for every $g\in G$, then $A=0$
when $r,s$ are inequivalent, and $A=cI_E$ for some $c\in\C$ when $E=F$ and $r=s$.
\end{lemma}

A partition $\lambda\vdash N$ is a sequence
$\lambda_1\ge\cdots\ge\lambda_s>0$ with $\sum_i\lambda_i=N$.
Its Young diagram has $\lambda_i$ boxes in row $i$, and
$\ell(\lambda)=s$ is its number of nonempty rows.
A standard tableau fills these boxes with $1,\ldots,N$, each once,
strictly increasingly along rows and columns.
A semistandard tableau with alphabet $[m]$ fills them with entries in
$[m]$, weakly increasing along rows and strictly increasing down columns.

Next we recall the Schur-Weyl decomposition as stated in \cite[Sections~2.1-2.3]{burchardt2025high} and \cite[Corollary 5.19.2, Theorem 5.22.1]{etingof2011introduction}.

\begin{lemma}[Schur--Weyl decomposition and bases]\label{lem:rt-sw}
For $m,N\ge1$, let
$W_m(\pi)(v_1\otimes\cdots\otimes v_N)
=v_{\pi^{-1}(1)}\otimes\cdots\otimes v_{\pi^{-1}(N)}$
for $\pi$ in the symmetric group $\mathfrak S_N$.
There is a unitary change of basis
\[
\Sch_m:(\C^m)^{\otimes N}\longrightarrow
\bigoplus_{\substack{\lambda\vdash N\\\ell(\lambda)\le m}}
\mathcal Q_\lambda^{(m)}\otimes\mathcal P_\lambda
\]
such that, for $U\in U(m)$ and $\pi\in\mathfrak S_N$,
\[
\Sch_mU^{\otimes N}\Sch_m^{\dagger}
=\bigoplus_\lambda u_\lambda^{(m)}(U)\otimes I_{\mathcal P_\lambda},\qquad
\Sch_mW_m(\pi)\Sch_m^{\dagger}
=\bigoplus_\lambda I_{\mathcal Q_\lambda^{(m)}}\otimes p_\lambda(\pi).
\]
The $u_\lambda^{(m)}$ are pairwise inequivalent irreducible representations
of $U(m)$; the $p_\lambda$ are pairwise inequivalent irreducible
representations of $\mathfrak S_N$, independent of $m$.
\end{lemma}
Let $\{\ket{\psi_{\lambda,a,b}^{(m)}}\}$
be the orthonormal basis of $(\mathbb C^m)^{\otimes N}$ associated with the Schur–Weyl decomposition above, where $a$ indexes the chosen basis of $\mathcal Q_\lambda^{(m)}$ and $b$ indexes the chosen basis of $\mathcal P_\lambda$. Define the \emph{Schur transform} $\Sch_m$ by
\[
\Sch_m:=\sum_{\lambda,a,b}\ket{\lambda, a, b}\bra{\psi_{\lambda,a,b}^{(m)}}.
\]
For an operator $X$ on $(\C^m)^{\otimes N}\otimes B$, define its diagonal block in these coordinates by
\[
X_{\lambda\lambda}:=(\bra{\lambda}\otimes I)(\Sch_m\otimes I_B)X(\Sch_m^\dagger\otimes I_B)(\ket{\lambda}\otimes I).
\]
For a state $X$, the probability of outcome $\lambda$ is $\Tr X_{\lambda\lambda}$, and, when positive, the conditional state on $\mathcal Q_\lambda^{(m)}\otimes\mathcal P_\lambda \otimes B$ is $X_{\lambda\lambda}/\Tr X_{\lambda\lambda}$. We recall the following fact about the Haar measure.

\begin{lemma}[Haar averaging {\citep[Section~II.C]{bartlett2007reference}}]\label{lem:rt-haar}
For every operator $X$ on $(\C^m)^{\otimes N}\otimes B$, %
\[
\begin{aligned}
&(\Sch_m\otimes I_B)
\left[\left(\cT_m^{(N)}\otimes I_B\right)(X)\right]
(\Sch_m^{\dagger}\otimes I_B)=
\bigoplus_{\substack{\lambda\vdash N\\\ell(\lambda)\le m}}
\frac{I_{\mathcal Q_\lambda^{(m)}}}{\dim\mathcal Q_\lambda^{(m)}}
\otimes\Tr_{\mathcal Q_\lambda^{(m)}}X_{\lambda\lambda}
\end{aligned}
\]
where $\cT_m^{(N)}(X):=\int_{U(m)}U^{\otimes N}X(U^\dagger )^{\otimes N}dU$ and $dU$ is the Haar measure.
\end{lemma}
Finally, 
\citet{burchardt2025high,bacon2005quantum}
show that the Schur transform can be implemented efficiently. The exact
basis used for the transform is irrelevant for us, although we are going to use
the fact that the basis used for the permutation matrix is independent of the local dimension $m$ and the basis for $\mathcal Q_\lambda^{(m)}$ 
consists of semistandard tableaux.

\begin{lemma}[Schur circuits]\label{lem:rt-circuits}
For $0<\epsilon<1$ and the common tableau bases above, $\Sch_m$ and $\Sch_m^{\dagger}$ admit circuits
of size $\poly(N,\log m,\log(1/\epsilon))$ with operator-norm error
at most $\epsilon$, on the valid input subspaces with auxiliary registers
initialized to zero.
\end{lemma}

\section{Implementation of the Quantum Channels from \Cref{sec:channels}}\label{app:subroutines}
\lemmadrfilter*
\begin{proof}
Start with two copies of $\rho$, perform the standard swap test, and retain the control register and the $j$-th factor of the first copy. Then append a qubit prepared in $\ket 1$. This procedure prepares
\[
\frac12\sum_{b\in\{0,1\}} \ket b\bra b \otimes \ket 1\bra 1\otimes \Tr_{-j}(\rho+(-1)^b\rho^2).
\]
Observe that the difference between the two subnormalized blocks for $b=0$ and $b=1$ is $\ket 1 \bra 1\otimes \Tr_{-j}(\rho^2)=\ket 1 \bra 1\otimes A_{j}$.

This is convenient because we can use a result of \citet[Lemma~12]{kimmel2017hamiltonian} which states that for a state $\rho'\in \cD(\C^2\otimes (\C^d)^{\otimes k})$ of the form $\rho'=\ket 0 \bra 0 \rho_++\ket 1 \bra 1 \rho_-$, we can simulate the Hamiltonian $H=\rho_+-\rho_-$, with trace-norm error $\epsilon$, using $O(\epsilon^{-1})$ copies of $\rho'$. We apply this to the subnormalized states $\rho_+=\frac12\ket 1\bra 1\otimes \Tr_{-j}(\rho+\rho^2)$ and $\rho_-=\frac12\ket 1\bra 1\otimes \Tr_{-j}(\rho-\rho^2)$ which implements the two unitaries
\[
e^{\pm iH/2} = e^{\pm i (\ket 1 \bra 1\otimes A_{j})/2} = \ket 0 \bra 0\otimes I+\ket 1 \bra 1\otimes e^{\pm i A_j/2},
\]
up to error $\epsilon$ using $O(\epsilon^{-1})$ copies in time $\poly(\epsilon^{-1}k\log d)$. So now we can assume that we have access to such transformations.

We describe the algorithm implementing $\Phi_j$ next. Let $L=O(k^2/\eta^4)$ be the Lipschitzness of $f$ and let $n=\Theta(L^2/\delta^2)$. We then append $m=O(n\log n)$ extra qubits: $n$ work qubits, an $M=O(\log n)$-qubit count register, and an extra rotation qubit for a total of $m=n+M+1$ extra qubits.

\begin{enumerate}
\item To each of the $n$ working qubits, apply the following transformations
\[
\textnormal{Hadamard}\to \ket 0 \bra 0\otimes I+\ket 1 \bra 1\otimes e^{+ i A_j/2}\to \textnormal{diag}(1,i)\to \textnormal{Hadamard},
\]
where, when we apply $\ket 0 \bra 0\otimes I+\ket 1 \bra 1\otimes e^{\pm i A_j/2}$, we mean applying it to the work qubit we are considering, and the target register $\cong \C^d$ on which we want $\Phi_j$ to act;

\item Compute the Hamming weight of the work qubits and write it into the counter register. In terms of basis
\[
\ket{b_1\cdots b_n}\ket 0 \to \ket{b_1\cdots b_n}\ket {b_1+\cdots +b_n};
\]

\item For $\ell\in\{0,\ldots, n\}$ define $\widehat\lambda(\ell)=2\arcsin\left(\left[\frac{2\ell}n-1\right]_0^{\sin(\frac12)}\right)$, where $[\cdot]_{a}^b$ is the capping between $a$ and $b$.
Then apply the following transformation to the counter register and rotation qubit as
\[
\ket\ell\ket 0\to \ket\ell \left( f(\widehat\lambda(\ell))\ket 0+\sqrt{1-f(\widehat\lambda(\ell))^2}\ket 1\right);
\]

\item Now we clean all appended registers. First, for the work qubits and the counter register:
\[
\ket{b_1\cdots b_n}\ket {b_1+\cdots +b_n}\to \ket{b_1\cdots b_n}\ket 0,
\]
and then reverse all the operations of Step~1 for each work qubit:
\[
\textnormal{Hadamard}\to \textnormal{diag}(1,-i)\to \ket 0 \bra 0\otimes I+\ket 1 \bra 1\otimes e^{- i A_j/2}\to \textnormal{Hadamard}.
\]
The rotation qubit is left untouched.

\item Measure the appended registers with the projector $\{\ket{0}\bra{0}^{\otimes m}, I-\ket{0}\bra{0}^{\otimes m}\}$. If we measure $\ket{0^m}$, then we retain the target register; otherwise, we output $\ket{\bot}$.
\end{enumerate}

Now we show that when the algorithm above is implemented with total error $O(\delta)$, the resulting channel $\widehat \Phi_j$ approximates the ideal channel $\Phi_j$ up to error $O(\delta)$. Assume that all gates are implemented exactly and fix a unit eigenvector $\ket \lambda\in \C^d$ of $A_j$, i.e., $A_j\ket\lambda=\lambda\ket\lambda$. Next we compute, on input $\ket\lambda$ in the target register, the state of each of the $n$ work qubits after Step~1. The joint state of the target register and a single work qubit is $\ket 0 \ket\lambda$, and after the Hadamard gate is $\frac1{\sqrt 2}(\ket 0\ket\lambda+\ket 1\ket\lambda)$. Then we apply the evolution given by the Hamiltonian simulation $\ket 0 \bra 0 \otimes I +\ket 1 \bra 1\otimes e^{iA_j/2}$, and the resulting state is %
\[
\frac1{\sqrt 2}\left(\ket 0+\ket 1e^{i A_j/2}\right)\ket \lambda = \frac1{\sqrt 2}\left(\ket 0+\ket 1e^{i \lambda/2}\right)\ket \lambda,
\]
and, after applying $\textnormal{diag}(1,i)$ we obtain
\[
\frac1{\sqrt 2}\left(\ket 0+\ket 1ie^{i \lambda/2}\right)\ket \lambda.
\]
Finally, after the final Hadamard gate, we obtain a working qubit in state $\alpha_\lambda\ket 0+\beta_\lambda\ket 1$, where $\alpha_\lambda=\frac{1+ie^{i\lambda/2}}{2}$ and $\beta_\lambda=\frac{1-ie^{i\lambda/2}}{2}$. %

After Step~2 and Step~3, the state over the target register and all ancillary qubits will be in state
\[
\sum_{\ell=0}^n \sum_{\substack{b\in\{0,1\}^n:\\|b|=\ell}} \alpha_\lambda^{n-\ell}\beta_\lambda^\ell \ket{b_1\cdots b_n}\ket\ell \Big( f(\widehat\lambda(\ell))\ket 0+\sqrt{1-f(\widehat\lambda(\ell))^2}\ket 1\Big) \ket\lambda.
\]
Now, we uncompute all the extra qubits but the rotation. Uncomputing the counter register is trivial. For each of the $n$ work qubits, we define as $U_\lambda$ the entire operator defined in Step~1. We proved that $U_\lambda\ket 0=\alpha_\lambda\ket 0+\beta_\lambda \ket 1$, so $U_\lambda^\dagger$ maps $\ket 0$ to $\ket 0$ with amplitude $\bar{\alpha_\lambda}$ and maps $\ket 1$ to $\ket 0$ with amplitude $\bar{\beta_\lambda}$. For a basis vector $\ket{b_1\cdots b_n}$ of Hamming weight $T$, the amplitude that all the work qubits are $\ket 0$ is
\[
\bar{\alpha_\lambda}^{n-T}\bar{\beta_\lambda}^{T}.
\]
And similarly, the contribution of the same basis vector from Step~1 is ${\alpha_\lambda}^{n-T}{\beta_\lambda}^{T}$ and $f(\widehat \lambda_T)$ from Step~3. Thus the contribution of each basis vector of Hamming weight $T$ is  %
\[
{\alpha_\lambda}^{n-T}{\beta_\lambda}^{T} f(\widehat \lambda_T) \bar{\alpha_\lambda}^{n-T}\bar{\beta_\lambda}^{T} = f(\widehat \lambda_T) \left|\alpha_\lambda\right|^{2(n-T)}\left|\beta_\lambda\right|^{2T}.
\]
Therefore, the total successful amplitude is 
\begin{align*}
g(\lambda)&=\sum_{T=0}^n\binom nT f(\widehat \lambda_T) \left|\alpha_\lambda\right|^{2(n-T)}\left|\beta_\lambda\right|^{2T}\\
&=\sum_{T=0}^n\binom nT f(\widehat \lambda_T) (1-p_\lambda)^{(n-T)}p_\lambda^{T}\tag{$p_\lambda:=\frac{1}2(1+\sin(\lambda/2))$}\\
&=\mathbb{E}_T[f(\widehat\lambda_T)] \tag{$T\sim \text{Bin}(n,p_\lambda)$}.
\end{align*}

Now we show that $g(\lambda)$ approximates $f(\lambda)$. First, note that by Lipschitzness of $f$ we have:
\[
|f(\widehat\lambda_T)-f(\lambda)| \le L|\widehat\lambda_T-\lambda|\le O(L)\left|\frac Tn-p_\lambda\right|,
\]
where the second inequality comes from the $O(1)$-Lipschitzness of the clipping function and the $\arcsin$ function. Thus
\begin{align*}
    |g(\lambda)- f(\lambda)|&=|\mathbb{E}_T[f(\widehat \lambda_T)]- f(\lambda)|\\
    &\le \mathbb{E}_T[|f(\widehat \lambda_T)- f(\lambda)|]\\
    &\le O(L)\mathbb{E}_T\left[\left|\frac Tn-p_\lambda\right|\right]\\
    &\le O(L)\sqrt{\mathbb{E}_T\left[\left(\frac Tn-p_\lambda\right)^2\right]}\\
    &\le O(L)\sqrt{(1-\sin(\lambda/2))/n}=O(L/\sqrt n) = O(\delta).
\end{align*}

Now instead of focusing on an input state which is an eigenvector of $A_j$, write it according to its spectral decomposition $A_j=\sum_{\lambda}\lambda\ket\lambda\bra\lambda$. We proved that on every eigenvector $|g(\lambda)-f(\lambda)|=O(\delta)$. Now let $G:=g(A_j)=\sum_{\lambda} g(\lambda)\ket\lambda\bra\lambda$. Since $g(\lambda)\in[0,1]$ then $0\preceq G \preceq I$. Moreover recall that $F_j=f(A_j)=\sum_{\lambda}f(\lambda)\ket\lambda\bra\lambda$ and thus
\[
\|G-F_j\|_{op} \le \max_{\lambda}|g(\lambda)-f(\lambda)| = O(\delta).
\]
This proves that with exact gates, the channel implements $\Phi_j$ with error at most $O(\delta)$. For the final error bound, observe that there are $O(n)$ Hamiltonian-simulation calls, each implemented to error $O(\delta/n)$, giving total error $O(\delta)$, and the remaining gates contribute at most another $O(\delta)$. Therefore, the full implementation error is at most $O(\delta)$.
The algorithm uses at most $2n\cdot O(n/\delta)=O(n^2/\delta)
=\poly(k,\eta^{-1},\delta^{-1})$ auxiliary copies and its time is bounded by $\poly(k,\eta^{-1},\delta^{-1},\log d)$.
\end{proof}

\lemmadresimate*
\begin{proof}
    Start with two copies of $\rho$ and prepare a control qubit in state $\ket+$. Let $S$ be the swap operator and apply $\ket 0 \bra0\otimes I+\ket 1 \bra 1 \otimes S$ to the joint states $\ket{+}\bra{+}\otimes \rho\otimes \rho$. Then apply the filtering algorithm of \Cref{lem:dr-filter} to the $j$-th factor of the first copy with accuracy $\delta=O(\epsilon)$. Then measure the control qubit in the Hadamard basis and check whether the filtered output gets sent to $\ket{\bot}$. So there are effectively $4$ outcomes
    \[
    (+,\bot),\quad(-,\bot),\quad(+,\lnot \bot),\quad\text{and}\quad(-,\lnot \bot).
    \]
    Define
    \[
    Z=\begin{cases}
        +1&\text{if}\quad (+,\bot)\\
        -1&\text{if}\quad (-,\bot)\\
        0&\text{otherwise}
    \end{cases},
    \]
    so that $\mathbb{E}[Z] = \mathbb{P}(+,\bot)-\mathbb{P}(-,\bot)$. Run this procedure $m$ times and return $\widehat e_j=\frac1m\sum_{r\in[m]}Z_r$. We will prove that $\mathbb{E}[Z]$ approximates $e_j$, and thus by standard concentration that so does $\widehat e_j$. 

    Assume that the filter $\Phi_j$ is ideal. Before filtering, the joint state is
    \[
    \omega:=\frac{1}2\left(\ket{0}\bra{0}\otimes (\rho\otimes \rho)+\ket{0}\bra{1}\otimes (\rho\otimes \rho)S+\ket{1}\bra{0}\otimes S(\rho\otimes \rho)+\ket{1}\bra{1}\otimes (\rho\otimes \rho)\right).
    \]
    By definition of the ideal channel $\Phi_j$, the probability of measuring $\bot$ on input $X$ is $\Tr((I-F_j^2) X)$. Since $\ket+\bra+-\ket -\bra -=\ket0\bra1+\ket 1\bra 0$, we get
    \begin{align*}
    \mathbb{E}[Z]&=\Tr[((\ket+\bra+-\ket -\bra -)\otimes ((I-F_j^2)\otimes I_{-j}\otimes I))\omega]\\
    &=\frac12\Tr[((I-F_j^2)\otimes I_{-j}\otimes I)((\rho\otimes \rho) S+S(\rho\otimes \rho))]\\
    &=\Re\Tr[((I-F_j^2)\otimes I_{-j}\otimes I)S(\rho\otimes \rho) ]\tag{by cyclicity of the trace}.
    \end{align*}
    Now it is easy to check that $\Tr[(Y\otimes I)S(\rho\otimes \rho)]=\Tr(Y\rho^2)$ and thus, for $Y=(I-F_j^2)\otimes I_{-j}$, we get
    \[
    \mathbb{E}[Z] = \Tr[((I-F_j^2)\otimes I_{-j})\rho^2]=\Tr[(I-F_j^2)\Tr_{-j}(\rho^2)] = e_j.
    \]
    Since $Z_r\in[-1,1]$, by taking $m=\Theta(\epsilon^{-2}\log (k/p))$ we get that $|\widehat e_j-e_j|=|\widehat e_j-\mathbb{E}[Z]|$ is larger than $\epsilon$ with probability at most $p/k$. A union bound shows uniformity over all $j\in[k]$.

    Each repetition uses two target copies and $\poly(k,\eta^{-1},\epsilon^{-1})$ auxiliary copies. The controlled swap costs $O(k\log d)$ gates, and the filtering time is polynomial in $k,\eta^{-1},\epsilon^{-1}$ and $\log d$. Multiplying by the number of repetitions and parties gives the claimed copy and time bounds. Since the implemented channel is $\delta$-close in trace norm even with references, the final measurement distribution changes by total variation at most $O(\delta)$.
\end{proof}

\lemmadrschurcompression*
\begin{proof}
    Recall that the Schur-Weyl decomposition prescribes the existence of a unitary change of basis
    \[
    \Sch_q:(\C^q)^{\otimes N} \to \bigoplus_{\substack{\lambda\vdash N\\\ell(\lambda)\le q}} Q_\lambda^{(q)}\otimes P_\lambda.
    \]
    Fix some unit vector $\ket{t_\lambda}\in Q_{\lambda}^{(q)}$. We define $\cC_{D\to q}^{(N)}$ by the following procedure:
    \begin{enumerate}
    \item Apply $\Sch_D$ and measure the $\lambda$ register;
    \item If $\ell(\lambda)>q$, discard the Schur registers, output any state in $(\C^q)^{\otimes N}$;
    \item Otherwise retain $\lambda$ and $P_\lambda$, discard $Q_\lambda^{(D)}$, and prepare the unit basis vector $|t_\lambda\rangle\in Q_\lambda^{(q)}$;
    \item Apply $\Sch_q^\dagger$.
    \end{enumerate}

    Note that $V^{\otimes N}$ commutes with permutations $W_D(\pi)$ for every $\pi\in \mathfrak S_N$:
    \[
    W_D(\pi)V^{\otimes N} = V^{\otimes N}W_q(\pi).
    \]
    Applying \Cref{lem:rt-schur} to the $P$ factors in \Cref{lem:rt-sw} gives
    \[
    \begin{aligned}
    \Sch_DV^{\otimes N}\Sch_q^\dagger
    =\bigoplus_{\substack{\lambda\vdash N\\\ell(\lambda)\le q}}
    V_\lambda\otimes I_{ P_\lambda},\qquad
    V_\lambda: Q_\lambda^{(q)}\to Q_\lambda^{(D)},
    \qquad V_\lambda^\dagger V_\lambda=I,
    \end{aligned}
    \]
    where $V_\lambda^\dagger V_\lambda =I$ follows from $V^{\otimes N}$ being an isometry. Moreover, no Schur label with $\ell(\lambda)>q$ occurs, since such labels are absent from the Schur decomposition of $(\mathbb C^q)^{\otimes N}$.
    Let $X_{\lambda\lambda}$ denote the unnormalized diagonal blocks of $(\Sch_q\otimes I_B)X(\Sch_q^\dagger\otimes I_B)$. Note that after measuring $\lambda$ of the embedded state, the corresponding $\lambda$ diagonal block becomes $(V_\lambda\otimes I_{P_\lambda\otimes B})X_{\lambda\lambda}(V_\lambda^\dagger\otimes I_{P_\lambda\otimes B})$ and
    \[
    \Tr_{Q_\lambda^{(D)}}\left[(V_\lambda\otimes I_{P_\lambda\otimes B})X_{\lambda\lambda}(V_\lambda^\dagger\otimes I_{P_\lambda\otimes B})\right] = \Tr_{Q_\lambda^{(q)}} [X_{\lambda\lambda}],
    \]
    since $V_\lambda^\dagger V_\lambda = I$.
    This indicates that after disregarding the unknown $Q_\lambda^{(D)}$ space, we exactly preserved the marginal on $P_\lambda\otimes B$. Then we insert an arbitrary fixed $\ket{t_\lambda}$ into the $Q_\lambda^{(q)}$ registers, and thus, at the end of Step~3 we have the state
    \[
    \bigoplus_{\substack{\lambda\vdash N\\\ell(\lambda)\le q}}|t_\lambda\rangle\langle t_\lambda|\otimes\Tr_{ Q_\lambda^{(q)}}(X_{\lambda\lambda}),
    \]
    and thus, for $Y=(V^{\otimes N}\otimes I_{ B})X((V^\dagger)^{\otimes N}\otimes I_{B})$:
    \[
    (\cC_{D\to q}^{(N)}\otimes I_B)\left(Y\right) = (\Sch_q^\dagger\otimes  I_B)\left[\bigoplus_{\substack{\lambda\vdash N\\\ell(\lambda)\le q}}|t_\lambda\rangle\langle t_\lambda|\otimes\Tr_{ Q_\lambda^{(q)}}(X_{\lambda\lambda})\right](\Sch_q\otimes  I_B)
    \]

    Now we use \Cref{lem:rt-haar} after applying $\cT_q^{(N)}\otimes I_B$ to $(\cC_{D\to q}^{(N)}\otimes I_B)\left(Y\right)$.
    Indeed:
    \begin{align*}
        (\cT_q^{(N)}\otimes I_B) ((\cC_{D\to q}^{(N)}&\otimes I_B)\left(Y\right)) = (\cT_q^{(N)}\otimes I_B)\left((\Sch_q^\dagger\otimes  I_B)\left[\bigoplus_{\substack{\lambda\vdash N\\\ell(\lambda)\le q}}|t_\lambda\rangle\langle t_\lambda|\otimes\Tr_{ Q_\lambda^{(q)}}(X_{\lambda\lambda})\right](\Sch_q\otimes  I_B)\right)\\
        &=(\Sch_q^\dagger\otimes  I_B)\left(\bigoplus_{\substack{\lambda\vdash N\\\ell(\lambda)\le q}}\frac{I_{Q_\lambda^{(q)}}}{\dim Q_\lambda^{(q)}}\otimes \Tr_{Q_\lambda^{(q)}}\left(\ket{t_\lambda}\bra{t_\lambda}\otimes\Tr_{Q_\lambda^{(q)}}(X_{\lambda\lambda})\right)\right)(\Sch_q\otimes  I_B)\\
        &=(\Sch_q^\dagger\otimes  I_B)\left(\bigoplus_{\substack{\lambda\vdash N\\\ell(\lambda)\le q}}\frac{I_{Q_\lambda^{(q)}}}{\dim Q_\lambda^{(q)}}\otimes \Tr_{Q_\lambda^{(q)}}\left(X_{\lambda\lambda}\right)\right)(\Sch_q\otimes  I_B)\\
        &=(\cT_q^{(N)}\otimes I_B)(X),
    \end{align*}
    where the second equality comes from \Cref{lem:rt-haar} under the Schur-Weyl basis, the third equality follows from $\Tr(\ket{t_\lambda}\bra{t_\lambda})=1$, and the last equality is again \Cref{lem:rt-haar}, this time applied directly.
\end{proof}

\end{document}